\documentclass[letterpaper,twocolumn,10pt]{article}
\usepackage{USENIX}

\usepackage{tikz}
\usepackage{amsmath}

\usepackage{filecontents}

\usepackage{graphicx}
\usepackage{subfigure}
\usepackage{amsmath}
\usepackage{amsthm}
\usepackage{amsfonts,amssymb}
\usepackage{mathrsfs}
\usepackage[vlined,ruled,linesnumbered,algo2e]{algorithm2e}
\usepackage{xcolor}
\usepackage{multicol}
\usepackage{multirow}
\usepackage{diagbox}
\usepackage{bbm}
\usepackage{makecell}
\usepackage{enumitem}
\usepackage{booktabs}
\usepackage{algorithm}
\usepackage{algorithmic}
\usepackage{appendix}
\usepackage{hyperref}
\usepackage{colortbl}
\usepackage{tcolorbox}

\newtheorem{thm}{Theorem}

\newtheorem{lem}{Lemma}

\begin{document}

\title{Data Duplication: A Novel Multi-Purpose Attack Paradigm in Machine Unlearning}

\author{{\rm Dayong Ye$^*$} \hspace{2mm} {\rm Tianqing Zhu$^{\dagger}$} \hspace{2mm} {\rm Jiayang Li$^*$} \hspace{2mm} {\rm Kun Gao$^*$} \\ {\rm Bo Liu$^*$} \hspace{2mm} {\rm Leo Yu Zhang$^{\P}$} \hspace{2mm} {\rm Wanlei Zhou$^{\dagger}$} \hspace{2mm} {\rm Yang Zhang$^{\S}$}\\
$^*$University of Technology Sydney \hspace{2mm} $^{\dagger}$City University of Macau \\ $^{\P}$Griffith University \hspace{2mm} $^{\S}$CISPA Helmholtz Center for Information Security}






\pagestyle{plain} 

\maketitle

\renewcommand{\thefootnote}{}
\footnotetext{Tianqing Zhu is the corresponding author.}

\begin{abstract}
Duplication is a prevalent issue within datasets. Existing research has demonstrated that the presence of duplicated data in training datasets can significantly influence both model performance and data privacy. However, the impact of data duplication on the unlearning process remains largely unexplored. This paper addresses this gap by pioneering a comprehensive investigation into the role of data duplication, not only in standard machine unlearning but also in federated and reinforcement unlearning paradigms. 
Specifically, we propose an adversary who duplicates a subset of the target model's training set and incorporates it into the training set. After training, the adversary requests the model owner to unlearn this duplicated subset, and analyzes the impact on the unlearned model. For example, the adversary can challenge the model owner by revealing that, despite efforts to unlearn it, the influence of the duplicated subset remains in the model.
Moreover, to circumvent detection by de-duplication techniques, we propose three novel near-duplication methods for the adversary, each tailored to a specific unlearning paradigm. We then examine their impacts on the unlearning process when de-duplication techniques are applied. 
Our findings reveal several crucial insights: 1) the gold standard unlearning method, retraining from scratch, fails to effectively conduct unlearning under certain conditions; 2) unlearning duplicated data can lead to significant model degradation in specific scenarios; 
and 3) meticulously crafted duplicates can evade detection by de-duplication methods. The source code is provided at: https://zenodo.org/records/14736535. 
\end{abstract}

\vspace{-2mm}
\section{Introduction}
\vspace{-1mm}

Machine learning requires vast amounts of data to effectively train models for various applications, including image processing \cite{Yang23} and natural language processing \cite{Ouyang22}. The ultimate performance of these models is heavily dependent on the quality of the training data. Presently, training datasets have expanded significantly, ranging from gigabytes to terabytes in size \cite{Xue21ACL}. While large datasets contribute to enhanced model performance, they also introduce a potential risk – the duplication of training data.
Current research has shown that duplicated data can detrimentally impact the overall performance of trained models \cite{Lee22ACL}, introduce heightened privacy risks \cite{Carlini21USENIX}, or result in a plethora of false alarms \cite{Mu22NDSS}. However, one crucial research area that has been largely overlooked in data duplication is machine unlearning. 

The concept of machine unlearning originates from data protection regulations such as the General Data Protection Regulation (GDPR) \cite{GDPR}, which empowers users to request the removal of their data. Under these regulations, model owners are obligated to comply with users' requests, removing revoked data from their datasets and ensuring the elimination of any influence these revoked data may have on the model. Existing machine unlearning research primarily focuses on the exact or approximate unlearning of revoked data and the verification of unlearning results \cite{Xu23}.
However, the extent to which duplicated data affects machine unlearning, such as the verification of unlearning results or the performance of unlearned models, remains unexplored.
This paper initiates the investigation of data duplication in machine unlearning. 

This research is significant due to three reasons.

\begin{itemize}[leftmargin=*]
    \item \textbf{Challenge in Verification.} When unlearning is applied to one duplicated subset, typically at the request of an adversary, it raises significant concerns about the validity of unlearning verification. This is because the other duplicated subset may still remain within the training set.
    


    \item \textbf{Model Collapse.} When the duplicated data include key features essential to the training set, unlearning one duplicated subset may still lead to model collapse, even if the other duplicate remains in the training set. This occurs because the unlearning process can disrupt the model's ability to generalize from those key features.

    \item \textbf{Avoiding De-duplication.} Model owners may implement de-duplication techniques to identify and eliminate duplicate data from training sets. A significant challenge arises in developing strategies to circumvent these detection mechanisms while still posing threats to the unlearning process.

    

\end{itemize}


This paper systematically explores the influence of duplicate data on machine unlearning from three perspectives. Initially, we examine how the presence of completely duplicated data influences both the unlearning outcomes and the subsequent verification results. Then, we carefully design near-duplicate data to investigate whether these synthetically generated near-duplicates have a similar impact to completely duplicated data on machine unlearning. Finally, we introduce and assess the effectiveness of de-duplication techniques to determine if their application can mitigate the impact of both complete and near-duplicate data on the machine unlearning process. 
In summary, we make three main contributions:
\begin{itemize}[leftmargin=*]
\item 
This work represents a significant stride in exploring the impact of data duplication on machine unlearning, establishing a pioneering effort in this emerging research area. 
Our study contributes novel insights in this domain. These insights include the failure of gold standard unlearning approaches, observed model degradation, and the challenges associated with duplication detection.

\item 
To explore the implications of duplicate data in machine unlearning, we introduce a novel set of near-duplication methods. These methods generate near-duplicate data by minimizing the feature distance between duplicate instances and their original counterparts, while simultaneously maximizing the perceptual differences between them. 

\item 
Beyond a focus on standard machine unlearning, our research extends its scope to encompass diverse unlearning domains. This includes federated unlearning and reinforcement unlearning. 
\end{itemize}

\vspace{-4mm}
\section{Preliminary and Threat Model}
\vspace{-2mm}

\subsection{Machine Unlearning} 
Formally, consider the original training set of the model denoted as $\mathcal{D}_{train}$, comprised of two distinct subsets: the set of unlearned data $\mathcal{D}_u$ and the set of retained data $\mathcal{D}_r$, represented as $\mathcal{D}_{train} = \mathcal{D}_u \cup \mathcal{D}_r$. Let $F$ symbolize the model that was trained on this combined dataset, $\mathcal{D}_{train}$. The central objective of machine unlearning is to derive an unlearned model, denoted as $F_u$, by effectively eliminating the influence of $\mathcal{D}_u$. 

\vspace{2mm}
\noindent\textbf{Threat Model.} We presume that the training dataset $\mathcal{D}_{train}$ includes a subset $\mathcal{D}_V$, associated with a victim, targeted by an adversary. 
The adversary is assumed to have the capability to inject a set of samples, $\mathcal{D}_A$, into $\mathcal{D}_{train}$, resulting in $\mathcal{D}_{train}\leftarrow\mathcal{D}_{train}\cup\mathcal{D}_A$. This assumption aligns with conventional data poisoning attacks \cite{Shan22USENIX}. 
For example, an adversary could inject duplicated data into a public dataset and upload it to data repository platforms such as GitHub\footnote{https://github.com} or Hugging Face\footnote{https://huggingface.co}. Any model owner who uses this dataset for training is vulnerable to the adversary’s request to unlearn the duplicated data.

The injected data, $\mathcal{D}_A$, may be identical to $\mathcal{D}_V$ or functionally similar yet perceptually distinct to evade de-duplication techniques. Specifically, for each data point pair $x^A_i\in\mathcal{D}_A$ and $x^V_i\in\mathcal{D}_V$ where $i\in{1,...,m}$ and $m$ is the size of $\mathcal{D}_A$, we require that $d(F(x^A_i), F(x^V_i)) < \delta$, with $d()$ denoting a distance metric like the $l_2$ norm and $\delta$ a pre-set threshold. 

\vspace{2mm}
\noindent\textbf{Adversary's Goal.} The adversary's goal is to challenge the model owner's success in unlearning. To achieve this, the adversary requests the revocation of their duplicated data, $\mathcal{D}_A$. After the unlearning process, the adversary challenges the model owner to verify the success of the unlearning. If the original data, $\mathcal{D}_V$, remains in the model's training set, $\mathcal{D}_{train}$, this verification is highly likely to fail. In such a case, the adversary can claim that the model owner is dishonest.

\vspace{-2mm}
\subsection{Federated Unlearning} 
\vspace{-1mm}
Federated learning (FL) involves a server and a set of clients working collaboratively to iteratively train a global model. Formally, let the number of clients in an FL system be $n$, each with their own local training dataset $\mathcal{D}_i$, where $i\in\{1,...,n\}$. The joint training dataset across all clients is represented by $\mathcal{D}=\bigcup^n_{i=1}\mathcal{D}_i$. 
The objective of the clients is to collaboratively learn a shared global model by solving the optimization problem: $\mathbf{w}^*=\text{argmin}_{\mathbf{w}}f(\mathcal{D},\mathbf{w})$, where $\mathbf{w}^*$ represents the optimal global model and $f(\mathcal{D},\mathbf{w})$ denotes the empirical loss of the model on dataset $\mathcal{D}$. 
The following three steps are iteratively taken by the server and the clients to train the global model.

\begin{itemize}[leftmargin=*]
    \item Step 1: synchronization. The server sends the current global model $\mathbf{w}$ to the clients. 

    \item Step 2: local model training. Once a client receives the global model $\mathbf{w}$, it fine-tunes its local model using their respective local dataset. For instance, consider client $i$. It initializes its local model $\mathbf{w}^i$ with the global model $\mathbf{w}$ and uses stochastic gradient descent to optimize $\mathbf{w}^i$ by solving the optimization problem: $min_{\mathbf{w}^i}f(\mathcal{D}_i,\mathbf{w}^i)$. Then, the client uploads its local model update $\Delta^i=\mathbf{w}^i-\mathbf{w}$ to the server.
    
    \item Step 3: global model updating. The server aggregates the local model updates received from the clients to compute a global model update $\Delta$. A common aggregation rule used is \emph{FedAvg} \cite{McMahan17}, defined as: $\Delta=\frac{1}{n}\sum^n_{i=1}\frac{|\mathcal{D}_i|}{|\mathcal{D}|}\Delta^i$. The global model update $\Delta$ is then used to fine-tune the global model: $\mathbf{w}\leftarrow\mathbf{w}-\alpha\cdot\Delta$, where $\alpha$ represents the global learning rate.
\end{itemize}

In contrast to FL, federated unlearning focuses on the targeted removal of acquired knowledge from the global model $\mathbf{w}$ \cite{Zhao23arxiv,Liu24ACMSurvey}. Depending on the level of granularity for unlearning, current federated unlearning methodologies can be categorized into three distinct groups \cite{Wu22Network}: sample-level unlearning, class-level unlearning, and client-level unlearning. Sample-level unlearning is designed to selectively unlearn individual data samples \cite{Pan23ICLR,Dhasade23}, class-level unlearning seeks to erase an entire class from the trained global model \cite{Wang22WWW}, and client-level unlearning is dedicated to eliminating all locally contributed model updates from a specific client \cite{Liu21IWQoS,Anisa23}. This paper concentrates on client-level unlearning, as it presents unique challenges. 
In comparison, sample-level unlearning is more analogous to scenarios encountered in standard machine unlearning. Class-level unlearning, however, introduces slight differences. For instance, if an adversary duplicates a class of data from a single victim client, unlearning that class has minimal impact on the global model’s performance since other clients also contribute data from the same class. Conversely, if the adversary duplicates a class of data across all clients and requests unlearning of that class, the global model will lose all knowledge related to the class, rendering it unable to accurately classifying data from that class.

\vspace{2mm}
\noindent\textbf{Threat Model.} We, without loss of generality, presume that the $n$-th client is a malicious entity, i.e., an adversary, whose aim is to compromise the global model through unlearning. This adversary is envisioned to actively participate throughout the complete training process of the global model and possesses the ability to craft their local model updates. This assumption is consistent with existing FL model poisoning attacks \cite{Fang20USENIX,She21NDSS}.
In each training round $t$, instead of uploading genuine local model updates, the adversary either replicates the received global model by directly setting $\mathbf{w}^n_t = \mathbf{w}_t$, or modifies the global model to satisfy $d(\mathbf{w}^n_t, \mathbf{w}_t) < \delta$ to bypass de-duplication mechanisms, and then sends $\mathbf{w}^n_t$ to the server. The adversary can easily perform the model upload operation, as it is a standard functionality of the FL system.

\vspace{2mm}
\noindent\textbf{Adversary's Goal.} The adversary's objective is to compromise the entire global model. Assuming the training concludes after $T$ rounds, the adversary executes an attack by revoking all their models, $\mathbf{w}^n_1, \ldots, \mathbf{w}^n_T$. Since these models closely resemble the respective $\mathbf{w}_1, \ldots, \mathbf{w}_T$, unlearning $\mathbf{w}^n_1, \ldots, \mathbf{w}^n_T$ effectively equates to unlearning $\mathbf{w}_1, \ldots, \mathbf{w}_T$. This process can severely damage the entire unlearned global model, $\mathbf{w}_u$, as it results in the loss of all knowledge embedded in $\mathbf{w}_1, \ldots, \mathbf{w}_T$.

\vspace{-2mm}
\subsection{Reinforcement Unlearning} 
\vspace{-1mm}

Formally, a reinforcement learning (RL) environment is commonly formulated as $\mathcal{M}=\langle\mathcal{S},\mathcal{A},\mathcal{T},r\rangle$ \cite{Mnih15Nature}.
Here, $\mathcal{S}$ and $\mathcal{A}$ denote the state and action sets, respectively. $\mathcal{T}$ represents the transition function, and $r$ represents the reward function.
At each time step $t$, the agent, given the current environmental state $s_t\in\mathcal{S}$, selects an action $a_t\in\mathcal{A}$ based on its policy $\pi(s_t,a_t)$.
This action causes a transition in the environment from state $s_t$ to $s_{t+1}$ according to the transition function: $\mathcal{T}(s_{t+1}|s_t,a_t)$.
The agent then receives a reward $r_t(s_t,a_t)$, along with the next state $s_{t+1}$. 
This tuple of information, denoted as $(s_t,a_t,r_t(s_t,a_t),s_{t+1})$, is collected by the agent as an experience sample utilized to update its policy $\pi$.
Typically, the policy $\pi$ is implemented using a Q-function: $Q(s,a)$, estimating the accumulated reward the agent will attain in state $s$ by taking action $a$.
Formally, the Q-function is defined as: 
\begin{equation}\label{eq:Q}
    Q_{\pi}(s,a)=\mathbb{E}_{\pi}[\sum^{\infty}_{i=1}\gamma^i\cdot r(s_i,a_i)|s_i=s,a_i=a],
\end{equation}
where $\gamma$ represents the discount factor. 

In deep reinforcement learning, a neural network is employed to approximate the Q-function, denoted as $Q(s,a;\theta)$, where $\theta$ represents the weights of the neural network.
The neural network takes the state $s$ as input and produces a vector of Q-values as output, with each Q-value corresponding to an action $a$.
To learn the optimal values of $Q(s,a;\theta)$, the weights $\theta$ are updated using a mean squared error loss function $\mathcal{L}(\theta)$. 
\begin{equation}\label{eq:loss}
\begin{aligned}
    \mathcal{L}=\frac{1}{|B|}\sum_{e\in B}[(r(s_t,a_t)+\gamma\max_{a_{t+1}}Q(s_{t+1},a_{t+1};\theta)-Q(s_t,a_t;\theta))^2],
\end{aligned}
\end{equation}
where $e=(s_t,a_t,r(s_t,a_t),s_{t+1})$ is an experience sample showing a state transition, 
and $B$ consists of multiple experience samples used to train the neural network.

During the unlearning process, the objective is to eliminate the influence of a specific environment on the agent, i.e., ``forgetting an environment'', which is equivalent to ``performing deterioratively in that environment'' \cite{Ye25NDSS}. 
Formally, let us consider a set of $n$ learning environments: $(\mathcal{M}_1,\ldots,\mathcal{M}_n)$. Each $\mathcal{M}_i$ has the same state and action spaces but differs in state transition and reward functions. 
Consider the target environment to be unlearned as $\mathcal{M}_u = \langle \mathcal{S}_u, \mathcal{A}_u, \mathcal{T}_u, r \rangle$, denoted as the `unlearning environment'. The set of remaining environments, denoted as $(\mathcal{M}_1, \ldots, \mathcal{M}_{u-1}, \mathcal{M}_{u+1}, \ldots, \mathcal{M}_n)$, will be referred to as the `retaining environments'.  
Given a learned policy $\pi$, the goal of unlearning is to update the policy $\pi$ to $\pi'$ such that the accumulated reward obtained in $\mathcal{M}_u$ is minimized: 
\begin{equation}\label{eq:aim}
    \min_{\pi'}||Q_{\pi'}(s)||_{\infty},
\end{equation}
where $s\in\mathcal{S}_u$, 
while the accumulated reward received in the retaining environments remains the same: 
\begin{equation}\label{eq:constraint}
    \min_{\pi'}||Q_{\pi'}(s)-Q_{\pi}(s)||_{\infty}, 
\end{equation} 
where $s\notin\mathcal{S}_u$.

\vspace{2mm}
\noindent\textbf{Threat Model.} We assume that an adversary can create an environment, $\mathcal{M}_A$, that is either identical or very similar to a given victim environment, $\mathcal{M}_V$, with a small difference $d(\mathcal{M}_A, \mathcal{M}_V) < \delta$ to evade de-duplication techniques. This assumption is practical, particularly in cases where the adversary is an internal attacker with regular access to $\mathcal{M}_V$. 
For example, an adversary can upload duplicated environments to shared repositories, such as simulation platform hubs like OpenAI Gym\footnote{https://github.com/openai/gym}. If the agent's owner uses these simulation platforms for training, they are exposed to the adversary’s subsequent request to unlearn the duplicated environments.

\vspace{2mm}
\noindent\textbf{Adversary's Goal.} The adversary's goal is to degrade the agent's performance in the victim environment, $\mathcal{M}_V$. To accomplish this, the adversary requests the model owner to erase their environment, $\mathcal{M}_A$. Due to the similarity between $\mathcal{M}_A$ and $\mathcal{M}_V$, unlearning $\mathcal{M}_A$ can inadvertently lead to the unlearning of $\mathcal{M}_V$. Consequently, this may result in a decline in the performance of the unlearned policy, $\pi'$, within $\mathcal{M}_V$.




\vspace{-2mm}
\section{Methodology}
\vspace{-1mm}
The central challenge when developing near-duplication methods lies in the creation of ``similar'' data, model updates, or environments, as the use of identical elements can be easily detected through de-duplication methods. Therefore, for each of the three learning settings, we will design a novel near-duplication method to effectively compromise the unlearning process while evading established de-duplication techniques. 


    

\vspace{-2mm}
\subsection{Duplication in Machine Unlearning}
\vspace{-1mm}
\noindent\textbf{Problem Statement.} Given a target model $F$, its training dataset $\mathcal{D}_{train}$, and a victim subset $\mathcal{D}_V$, 
where $\mathcal{D}_V\subset\mathcal{D}_{train}$, the adversary constructs a duplicated dataset $\mathcal{D}_A$ and inject it into $\mathcal{D}_{train}$, i.e., $\mathcal{D}_{train}\leftarrow\mathcal{D}_{train}\cup\mathcal{D}_A$. This duplicated dataset must match the size of $\mathcal{D}_V$ and exhibit functional similarity to $\mathcal{D}_V$.
Formally, for each pair of data points $x^A_i\in\mathcal{D}_A$ and $x^V_i\in\mathcal{D}_V$, the adversary aims to optimize two distinct losses: a utility loss (Eq. \ref{eq:utilityloss}) and a perceptual loss (Eq. \ref{eq:perceptualloss}).

\begin{equation}\label{eq:utilityloss}
    \mathcal{L}_u=min||E(x^A_i)-E(x^V_i)||_1,
\end{equation}
where $E$ is a feature extractor.
\begin{equation}\label{eq:perceptualloss}
    \mathcal{L}_p=max||x^A_i-x^V_i||_1.
\end{equation}

Here, Eq. \ref{eq:perceptualloss} is experimentally designed to balance overall performance and computational efficiency in creating near-duplicates. We have also evaluated alternative perceptual metrics such as SSIM \cite{Wang04TIP}. However, SSIM presented challenges in achieving convergence during the creation process.

\vspace{2mm}
\noindent\textbf{Overview.} The proposed duplication method against standard machine unlearning comprises three distinct steps.
First, the adversary initiates the process by generating a dataset $\mathcal{D}_{dup}$ through the direct duplication of $\mathcal{D}_V$, effectively rendering $\mathcal{D}_{dup} = \mathcal{D}_V$.
Second, the adversary employs $\mathcal{D}_{dup}$ to train an autoencoder $AE$, optimizing for the losses defined in both Eq. \ref{eq:utilityloss} and \ref{eq:perceptualloss}.
Lastly, the adversary applies the trained $AE$ to each sample in $\mathcal{D}_{dup}$ to yield $\mathcal{D}_A$. 
To carry out an attack, the adversary requests the server to unlearn the data $\mathcal{D}_A$.

\vspace{2mm}
\noindent\textbf{Details of the Method.} The crucial phase in the proposed method is the second step: training an autoencoder $AE$, shown in Figure \ref{fig:autoencoder}. 
This autoencoder takes a sample $x^V$ as input and produces a new sample $x^A$. 
It is expected that $||E(x^A) - E(x^V)||_1 < \delta_u$, while $||x^A - x^V||_1 > \delta_p$, where both $\delta_u$ and $\delta_p$ are pre-defined thresholds.

\begin{figure}[ht]
\vspace{-2mm}
\centering
	\includegraphics[scale=0.32]{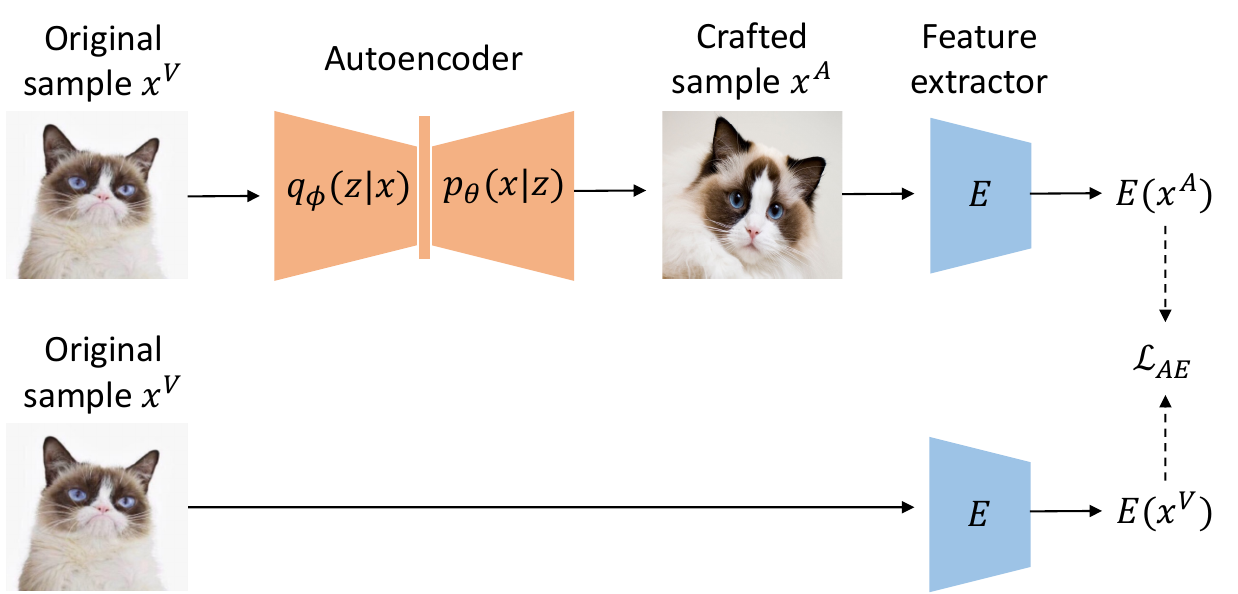}
 \vspace{-2mm}
	\caption{Training of the autoencoder.} 
	\label{fig:autoencoder}
 \vspace{-2mm}
\end{figure}

To train the autoencoder, the adversary utilizes $\mathcal{D}_{dup}$ as the training dataset, and the loss function is defined as:

\begin{equation}\label{eq:AEloss}
\mathcal{L}_{AE} = \min\left(||E(x^A) - E(x^V)||_1 - \lambda||x^A - x^V||_1\right),
\end{equation}
where $\lambda$ is used to balance the two terms. For example, in the case of an image classifier, the adversary can employ a publicly pretrained image classifier, such as MobileNetV2 \cite{Sandler18CVPR}, as a feature extractor by taking the output of its intermediate layers.
This approach aims to minimize the discrepancy between the extracted features of the original sample $x^V$ and the crafted sample $x^A$ while ensuring that the perturbation between $x^V$ and $x^A$ remains as large as possible.

\vspace{2mm}
\noindent\textbf{Analysis of the Method.} In this method, we trained an autoencoder to generate samples $\mathcal{D}_A$ that exhibit functional similarity to those in $\mathcal{D}_V$. The autoencoder comprises two integral components: an encoder and a decoder. The encoder provides the decoder with an approximation of its posterior over latent variables. Conversely, the decoder serves as a framework for the encoder to learn meaningful representations of the data. 

Formally, let $q_{\phi}(z|x)$ represent the encoder, and $p_{\theta}(x|z)$ represent the decoder, where $x$ and $z$ denote a sample and a latent variable, respectively. The parameters of the encoder and decoder are denoted by $\phi$ and $\theta$, respectively.
In essence, the autoencoder learns a stochastic mapping from the observed sample space $x$ to the latent space $z$, while the decoder learns an inverse mapping from the latent space to the sample space. Following Bayesian rules, the decoder $p_{\theta}(x|z)$ can be expressed as $p_{\theta}(x|z) = \frac{p_{\theta}(z|x)p_{\theta}(x)}{p_{\theta}(z)}$. Since $p_{\theta}(z|x)$ is the inverse of $p_{\theta}(x|z)$, it can be approximated by $q_{\phi}(z|x)$ \cite{Kingma19}. Consequently, we derive $p_{\theta}(x|z) = \frac{q_{\phi}(z|x)p_{\theta}(x)}{p_{\theta}(z)}$.

Through rearrangement and utilizing log-likelihood, we obtain the expression $log[p_{\theta}(x|z)] - log[p_{\theta}(x)] = log[q_{\phi}(z|x)] - log[p_{\theta}(z)]$. Notably, $log[p_{\theta}(x|z)] - log[p_{\theta}(x)]$ signifies the difference between the input and output of the autoencoder, identical to the disparity between the learned latent variable distribution by the encoder and the true latent variable distribution, i.e., $log[q_{\phi}(z|x)] - log[p_{\theta}(z)]$. As the true latent variable distribution $p_{\theta}(z)$ remains constant, controlling the difference between the input and output of the autoencoder necessitates appropriate training of the encoder.

As expressed in Eq. \ref{eq:AEloss}, the first term $||E(x^A) - E(x^V)||_1$ uses a well-trained feature extractor $E$ to guide the encoder's training toward the provided feature extractor. Also, as the objective is to craft perceptually distinct samples, the second term in Eq. \ref{eq:AEloss}, $||x^A - x^V||_1$, is incorporated to guide the decoder's training toward generating perceptually diverse samples. This dual-term loss function ensures the convergence of the autoencoder while balancing functional similarity and perceptual dissimilarity in the generated samples.

\subsection{Duplication in Federated Unlearning}
\noindent\textbf{Problem Statement.} In the context where the server generates a global model denoted as $\mathbf{w}$ and distributes it to all clients, the aim of the adversary, i.e., the malicious client, is to create a local model denoted as $\mathbf{w}^A$. This local model should have the same structure as $\mathbf{w}$ and exhibit functional similarity. However, the parameter values of $\mathbf{w}^A$ should differ from those of $\mathbf{w}$. Formally, given the adversary's local data $\mathcal{D}_A$, the goal is to optimize two separate losses: the similarity loss (Eq. \ref{eq:similarityloss}) and the model-distance loss (Eq. \ref{eq:modelloss}).

\begin{equation}\label{eq:similarityloss}
    \mathcal{L}_s=min||\mathbf{w}(x)-\mathbf{w}^A(x)||_1,
\end{equation}
where $\mathbf{w}()$ and $\mathbf{w}^A()$ represent the logits of the global model and the adversary's local model, respectively, while $x\in\mathcal{D}_A$ denotes a local sample from the adversary. 
\begin{equation}\label{eq:modelloss}
    \mathcal{L}_m=max||\mathbf{w}-\mathbf{w}^A||_1.
\end{equation}

Here, Eq. \ref{eq:modelloss} employs the $L_1$-norm as an intuitive method for comparing model parameters during the creation of near-duplicate models. This approach ensures that the generation process is efficient while preserving the fidelity of the duplicates by promoting small and consistent changes across the parameters, thereby maintaining the model's overall structure and functionality.

\vspace{2mm}
\noindent\textbf{Overview.} The duplication method consists of three steps. Firstly, upon receiving the global model $\mathbf{w}$, the adversary initializes the local model $\mathbf{w}^A$ by adopting the same structure as $\mathbf{w}$ and randomly initializing the parameter values of $\mathbf{w}^A$. Secondly, the adversary utilizes the local data $\mathcal{D}_A$ to train $\mathbf{w}^A$, optimizing the losses defined in both Eq. \ref{eq:similarityloss} and \ref{eq:modelloss}. Finally, the adversary transmits the trained local model $\mathbf{w}^A$ to the server. To execute an attack, the adversary then requests the server to unlearn their local model $\mathbf{w}^A$.

\vspace{2mm}
\noindent\textbf{Details of the Method.} The pivotal phase of this method is the second step: local model training. This training leverages the knowledge distillation technique, a process in which a student model is trained under the guidance of a teacher model \cite{Gou21}. In our scenario, the global model serves as the teacher, while the adversary's local model represents the student. However, in contrast to conventional knowledge distillation objectives, the adversary's intention is to train the local model to emulate the global model rather than achieving high performance. 
To achieve the adversary's objective, a direct approach involves defining a loss function that combines Eq. \ref{eq:similarityloss} and \ref{eq:modelloss}, namely $\min\left(||\mathbf{w}(x) - \mathbf{w}^A(x)||_1 - ||\mathbf{w}-\mathbf{w}^A||_1\right)$. However, Eq. \ref{eq:similarityloss} solely compares the outputs of the two models without accounting for their internal representations. Thus, an extension of Eq. \ref{eq:similarityloss} is necessary to also transfer the internal representations of the global model to the adversary's local model.

\begin{figure}[ht]
\vspace{-1mm}
\centering
	\includegraphics[scale=0.3]{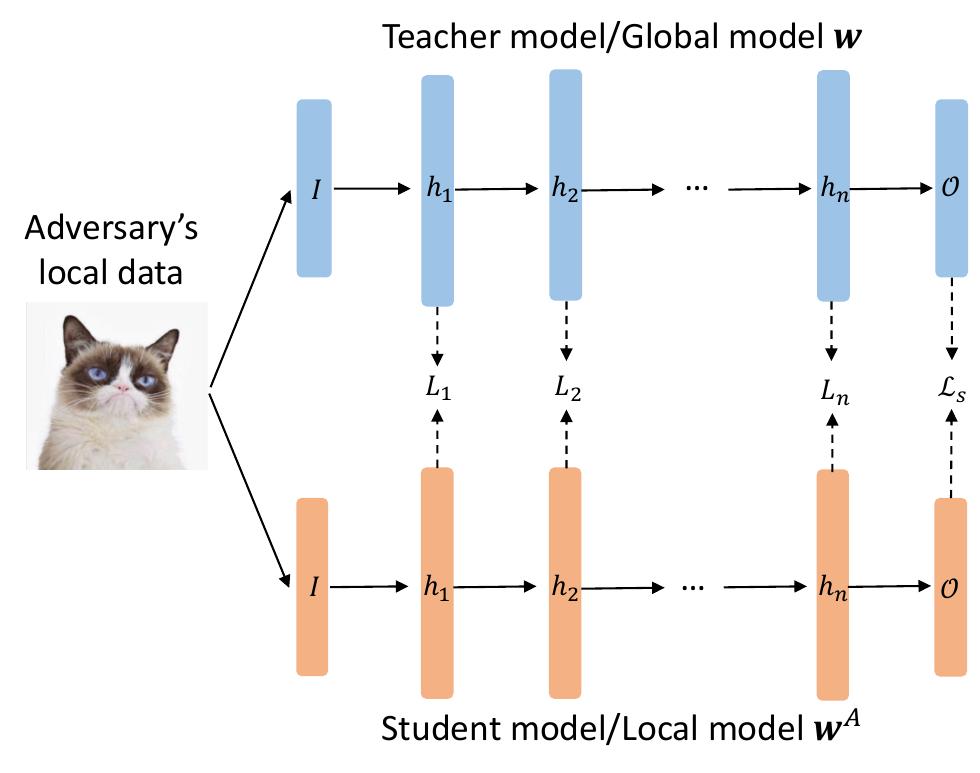}
 \vspace{-1mm}
	\caption{Training of the adversary's local model.} 
	\label{fig:KD}
 \vspace{-1mm}
\end{figure}

As depicted in Figure \ref{fig:KD}, for the student to mimic the teacher internally, we consider two scenarios. The first scenario is the internal distillation of all layers, where every layer of the student is fine-tuned to align with the corresponding layer in the teacher. The second scenario is the internal distillation of selected layers, where a selected subset of layers in the student is optimized to closely align with their counterparts in the teacher model. 


In our specific problem, we opt for the approach of internal distillation of selected layers. This deliberate choice stems from the recognition that distilling knowledge from all layers might result in the student model closely mirroring the teacher model. While this alignment is advantageous for knowledge transfer, it simultaneously heightens the risk of detection by de-duplication techniques. 
By distilling knowledge from specific layers, we guide the student model to internalize the most relevant and informative aspects of the teacher's knowledge while avoiding an overly conspicuous resemblance. Specifically, the loss function used by the adversary to train their local model, i.e., the student model, is defined as:

\begin{equation}\label{eq:FLloss}
\begin{aligned}
    &\mathcal{L}_{KD}=min(\lambda_1||\mathbf{w}_{(1)}(x)-\mathbf{w}_{(1)}^A(x)||_1+...+\\
    &\lambda_m||\mathbf{w}_{(m)}(x)-\mathbf{w}_{(m)}^A(x)||_1+||\mathbf{w}(x)-\mathbf{w}^A(x)||_1-\lambda||\mathbf{w}-\mathbf{w}^A||_1),
\end{aligned}
\end{equation}
where $\mathbf{w}_{(i)}()$ is the output of the $i$-th hidden layer of the model $\mathbf{w}$, $m$ is the number of hidden layers selected by the adversary, and $\lambda$, $\lambda_1,...,\lambda_m$ are used to balance these learning objectives. 

\vspace{2mm}
\noindent\textbf{Analysis of the Method.}  
The conventional knowledge distillation loss incorporates both soft and hard labels in a combined form: $\mathcal{L}=\rho L(\mathbf{y}_s,\mathbf{y}_t)+(1-\rho)L(\mathbf{y}_s,\mathbf{y}_g)$, where $\mathbf{y}_s$ and $\mathbf{y}_t$ represent the soft labels of the student and teacher, respectively, $\mathbf{y}_g$ denotes the ground truth label, and $\rho$ is referred to as the soft ratio. The efficacy of this loss function in ensuring the convergence of the student has been established in \cite{Aguilar20,Ji20}. 

In contrast to the conventional knowledge distillation loss, as expressed in Eq. \ref{eq:FLloss}, our approach differs by excluding the hard-label component $L(\mathbf{y}_s, \mathbf{y}_g)$. Our focus is on emulating the behavior of the teacher rather than enhancing the student's standalone performance. Instead, our strategy involves multiple objective losses pertaining to the internal representations of the teacher.
Given the proven convergence of $\mathcal{L}$, if our objective is to drive $L(\mathbf{y}_s, \mathbf{y}_t)$ towards zero, we only need to ensure that $L(\mathbf{y}_s, \mathbf{y}_g)$ approaches zero. This alignment is easily achievable in our context by inducing the student model to overfit to the adversary's local data. Notably, overfitting poses no concern in our setting, as our primary aim lies in the emulation of the teacher's behavior rather than the student model's generalization ability. Consequently, by inducing overfitting, we can effectively drive $L(\mathbf{y}_s, \mathbf{y}_t)$ towards zero.

Returning to Eq. \ref{eq:FLloss}, where $L(\mathbf{y}_s, \mathbf{y}_t)$ is equivalent to $||\mathbf{w}(x)-\mathbf{w}^A(x)||_1$, achieving convergence in $\mathcal{L}_{KD}$ necessitates balancing the values of the remaining terms: $||\mathbf{w}_{(1)}(x)-\mathbf{w}_{(1)}^A(x)||_1,...,||\mathbf{w}_{(m)}(x)-\mathbf{w}_{(m)}^A(x)||_1$ and $||\mathbf{w}-\mathbf{w}^A||_1$. A direct approach would be to induce all these terms to converge to $0$. However, such an approach implies $||\mathbf{w}-\mathbf{w}^A||_1\rightarrow 0$, which contradicts our primary goal. The critical step is judiciously selecting the value of $m$. A small $m$ may lead the loss to converge to sub-optimality, while a large $m$ can induce oscillations. 
We empirically set the value of $m$ to $2$ with one representing the output of the convolutional blocks and the other representing the output of the classification layer. 



\vspace{-2mm}
\subsection{Duplication in Reinforcement Unlearning}
\vspace{-1mm}

\noindent\textbf{Problem Statement.} In this context, when considering a victim environment $\mathcal{M}_V$, the adversary's objective is to construct a comparable environment, denoted as $\mathcal{M}_A$. This newly created environment, $\mathcal{M}_A$, should exhibit functional similarity to $\mathcal{M}_V$ but possess distinct parameter values. Formally, to establish functional equivalence, the agent operating under a learned policy $\pi$ should be able to achieve similar cumulative rewards in both $\mathcal{M}_A$ and $\mathcal{M}_V$, i.e.,
\begin{equation}\label{eq:reward}
    \mathcal{L}_r=min|Q_{\pi}(s,a|_{(s,a)\sim\mathcal{M}_V})-Q_{\pi}(s,a|_{(s,a)\sim\mathcal{M}_A})|.
\end{equation}
In pursuit of distinct parameter values, the adversary endeavors to alter $\mathcal{M}_V$ to $\mathcal{M}_A$. 
\begin{equation}\label{eq:transition}
    \mathcal{L}_t=max[d(\mathcal{M}_V,\mathcal{M}_A)].
\end{equation}

\vspace{2mm}
\noindent\textbf{Overview.} The duplication method targeting reinforcement unlearning is a strategic process that unfolds in three distinct steps. In the initial step, the adversary straightforwardly duplicates the victim environment $\mathcal{M}_V$. Subsequently, the adversary leverages $\mathcal{M}_V$ as a training ground to develop an environment generator. The training process of this generator is guided by the losses defined in both Eq. \ref{eq:reward} and \ref{eq:transition}. Lastly, the adversary deploys the trained generator to create a comparable environment $\mathcal{M}_A$. To execute the attack, the adversary simply instructs the agent to unlearn the generated environment $\mathcal{M}_A$.

\begin{figure}[ht]
\vspace{-2mm}
\centering
	\includegraphics[scale=0.35]{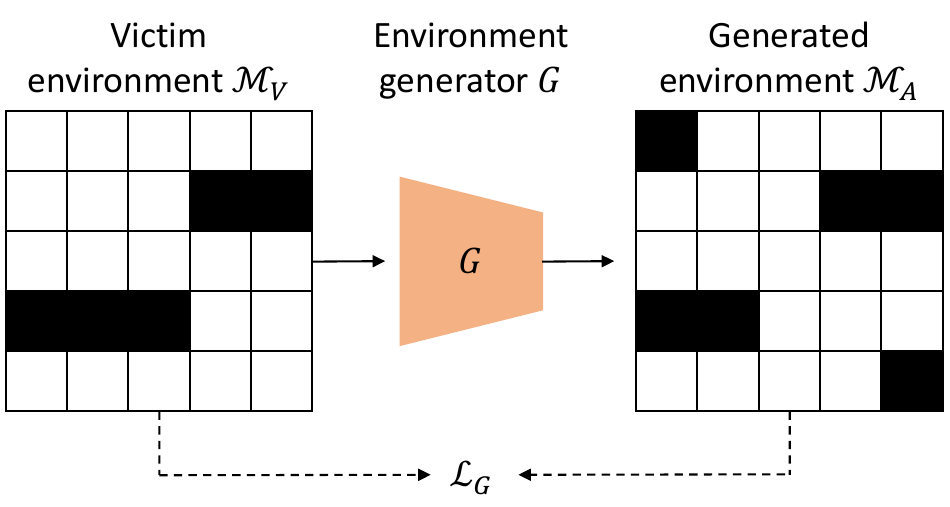}
 \vspace{-2mm}
	\caption{Training of the adversary's environment generator.} 
	\label{fig:EnvGen}
 \vspace{-2mm}
\end{figure}

\vspace{2mm}
\noindent\textbf{Details of the Method.} The pivotal phase of this method is the second step, which involves training an environment generator as illustrated in Figure \ref{fig:EnvGen}. The objective is to ensure that the resulting generated environment $\mathcal{M}_A$ possesses identical state and action spaces as $\mathcal{M}_V$, while eliciting a comparable cumulative reward for the agent under a learned policy $\pi$, namely $min(|\mathcal{L}_r-\mathcal{L}_t|)$. The differentiating factor lies solely in the state transition function, denoted as $\mathcal{T}_V$ for the victim environment and $\mathcal{T}_A$ for the generated environment. Technically, the adversary also has the capability to modify the state and action spaces of $\mathcal{M}_V$. Nevertheless, such modifications carry the potential to exert a significant impact on the overall environment.
This stems from the intricate relationship between the state and action spaces and the dynamics of RL environments. A change in the state space can result in altered observations available to the agent, influencing its perception of the environment. Similarly, a modification in the action space can redefine the set of feasible actions the agent can take, reshaping the decision-making process. Thus, both alterations can lead to a substantial transformation in the environment's behavior.

To guarantee that the generated environment $\mathcal{M}_A$ yields a similar cumulative reward for the agent as $\mathcal{M}_V$, the loss defined in Eq. \ref{eq:reward} is evaluated by using the learned policy $\pi$ to explore both $\mathcal{M}_A$ and $\mathcal{M}_V$, and subsequently comparing the acquired rewards. Formally, Eq. \ref{eq:reward} is implemented as:
\begin{equation}\label{eq:rewardR}
    \mathcal{L}_r=min|\sum^{T}_{i=1}r(s_i,a_i)_{(s_i,a_i)\sim\mathcal{M}_V}-\sum^{T}_{i=1}r(s_i,a_i)_{(s_i,a_i)\sim\mathcal{M}_A}|,
\end{equation}
where $T$ denotes the predefined number of steps taken to explore each of the two environments.
Concurrently, to achieve a distinct state transition function $\mathcal{T}_A$ within $\mathcal{M}_A$, the loss function specified in Eq. \ref{eq:transition} is implemented as: 
\begin{equation}\label{eq:transitionKL}
    \mathcal{L}_t=max[KL(\mathcal{T}_V(\cdot|s,a)||\mathcal{T}_A(\cdot|s,a))],
\end{equation}
where $KL(\cdot||\cdot)$ denotes the KL-divergence. 
We utilize KL-divergence as it is well-suited for generating near-duplicate transition functions represented as probability distributions. This metric guides the creation process by capturing subtle differences in probability distributions to ensure that the generated near-duplicates align with the desired properties.
The training loss of the environment generator is defined as: 
\begin{equation}\label{eq:generator}
    \mathcal{L}_G=min(|\mathcal{L}_r-\mathcal{L}_t|).
\end{equation}
Finally, attaining identical state and action spaces between $\mathcal{M}_A$ and $\mathcal{M}_V$ can be accomplished by properly defining the input and output of the environment generator.

\vspace{2mm}
\noindent\textbf{Analysis of the Method.} As our objective is to alter the state transition function of an environment, our analysis predominantly focuses on this function. 

Let $\mathcal{T}^{\pi}(s'|s)=\int_{\mathcal{A}}\pi(a|s)\mathcal{T}(s'|s,a){\rm d}a$, and the expected cumulative reward under policy $\pi$ be \cite{Metelli18}:
\begin{equation}\label{eq:cumulativeReward}
    J^{\mathcal{T},\pi}_{\mu}=\frac{1}{1-\gamma}\int_{\mathcal{S}}d^{\mathcal{T},\pi}_{\mu}(s)\int_{\mathcal{A}}\pi(a|s)r(s,a){\rm d}a{\rm d}s,
\end{equation}
where $d^{\mathcal{T},\pi}_{\mu}$ is the $\gamma$-discounted state distribution, recursively defined as: 
\begin{equation}
    d^{\mathcal{T},\pi}_{\mu}(s)=(1-\gamma)\mu(s)+\gamma\int_{\mathcal{S}}d^{\mathcal{T},\pi}_{\mu}(s')\mathcal{T}^{\pi}(s'|s){\rm d}s',
\end{equation}
and $\mu(s)$ is the distribution of the initial state. 
Then, we obtain:
\begin{lem}\label{lem:state distribution}
    Let $\mathcal{T}$ and $\mathcal{T}'$ be two distinct state transition functions. The $l_1$ norm of the difference between the $\gamma$-discounted state distributions can be upper bounded as:
    \begin{equation}\nonumber
        ||d^{\mathcal{T}',\pi}_{\mu}-d^{\mathcal{T},\pi}_{\mu}||_1\leq\frac{\gamma}{1-\gamma}\mathbb{E}_{s\sim d^{\mathcal{T},\pi}_{\mu}}||\mathcal{T}'^{\pi}(\cdot|s)-\mathcal{T}^{\pi}(\cdot|s)||_1.
    \end{equation}
\end{lem}

Based on Lemma \ref{lem:state distribution}, we can have:
\begin{thm}\label{thm:reward}
The difference in the cumulative rewards between two state transition functions, $\mathcal{T}'$ and $\mathcal{T}$, can be upper bounded as:
    \begin{equation}\nonumber
    \begin{aligned}
        &J^{\mathcal{T}',\pi}_{\mu}-J^{\mathcal{T},\pi}_{\mu}\leq\\
        &\frac{\gamma}{(1-\gamma)^2}\mathbb{E}_{s\sim d^{\mathcal{T},\pi}_{\mu}}||\mathcal{T}'^{\pi}(\cdot|s)-\mathcal{T}^{\pi}(\cdot|s)||_1\max_s\mathbb{E}_{a\sim\mathcal{A}}r(s,a).
    \end{aligned}
    \end{equation}
\end{thm}
Theorem \ref{thm:reward} indicates that the difference in rewards between two distinct state transition functions is upper-bounded by the difference between these two transition functions. Therefore, achieving a balance between $\mathcal{L}_r$ and $\mathcal{L}_t$ is crucial. 
In essence, a delicate equilibrium is sought, typically favoring a minimized $\mathcal{L}_r$ and a maximized $\mathcal{L}_t$ to fulfill the adversary's objectives. However, overly aggressive maximization of $\mathcal{L}_t$ may result in a substantial deviation from the original cumulative reward, potentially making the method more conspicuous and prone to detection mechanisms. 
The proofs of both Lemma \ref{lem:state distribution} and Theorem \ref{thm:reward} are provided in the appendix.

\section{Experiments}
\vspace{-2mm}
Our experimental focus is on machine unlearning and federated unlearning. As reinforcement unlearning differs significantly in methodology, its results are presented separately.

\vspace{-2mm}
\subsection{Experimental Setup}
\noindent\textbf{Evaluation Tasks.} Duplication can be executed in either a complete or a similar manner, while de-duplication techniques can be either applied or not applied. Here, complete duplication entails fully copying the unlearning data, models, or environments, while similar duplication involves copying only the features of those unlearning entities, as shown in our proposed methods. Thus, we set up four evaluation tasks. 
\begin{itemize}[leftmargin=*]
    \item \textit{Complete duplication without de-duplication.} This task evaluates the unlearning outcomes when the unlearning entities are completely copied without applying de-duplication techniques. It serves as a baseline for our evaluation.

    \item \textit{Similar duplication without de-duplication.} This task evaluates the performance of our proposed methods without applying any de-duplication techniques. It serves to demonstrate the upper bound performance of our methods.  
    
    \item \textit{Complete duplication with de-duplication.} This task evaluates the unlearning results under complete duplication when the de-duplication techniques are applied. It serves as the lower bound performance in our evaluation.

    \item \textit{Similar duplication with de-duplication.} This task is a crucial evaluation, assessing the performance of our proposed methods when confronted with adopted de-duplication techniques. It aims to determine whether our methods can effectively bypass these de-duplication techniques.
\end{itemize}

\noindent\textbf{De-duplication Techniques.} De-duplication techniques vary across diverse learning paradigms. 

\vspace{1mm}
\textit{Machine Unlearning.} 
We utilize feature-based de-duplication \cite{Pragash21}, which poses a significant challenge to our near-duplication method. Our method deliberately minimizes the feature distance between duplicates and their originals to evade detection, making feature-based de-duplication an effective countermeasure. For implementation, we employ a publicly available pre-trained image classifier, such as VGG16 \cite{Simonyan15}, as the feature extractor.

\vspace{1mm}
\textit{Federated Unlearning.} 
We adopt a combination of representational and functional measures \cite{Csis21NIPS,Klab23arxiv}, commonly used to assess model similarity. This dual approach ensures a strong countermeasure against our near-duplication method because it addresses both the internal representations and external behaviors of the model updates, leaving minimal room for near-duplicates to evade detection.

\vspace{1mm}
\textit{Reinforcement Unlearning.} 
We choose cosine similarity as the metric for detecting duplicates, as it aligns with the common representation of environments as state vectors. This intuitive approach counters our near-duplicates, which are generated by modifying individual grids within an environment, akin to altering specific elements within a vector.

\vspace{2mm}
\noindent\textbf{Unlearning Methods.} 
A commonly used unlearning method, denoted as \textbf{Retrain}, can be implemented across all these learning scenarios. Widely recognized as a gold standard in existing literature \cite{Zhang22NIPS} , the \textbf{Retrain} method involves retraining the model from scratch, excluding the unlearned entities.

In machine unlearning, we utilize two representative methods broadly employed in existing research. 

\begin{itemize}[leftmargin=*]

    \item \textbf{Fisher forgetting \cite{Golatkar20CVPR}.} This method applies additive Gaussian noise to perturb the model towards exact unlearning. The Gaussian distribution has a zero mean and covariance determined by the $4$-th root of the Fisher information matrix with respect to the model on the unlearned data.
    
    \item \textbf{Relabeling \cite{Graves21AAAI}.} This method operates on the unlearned data by altering the labels to randomly selected incorrect ones. Subsequently, these mislabeled data are utilized to fine-tune the model. 

\end{itemize}

Within the context of federated unlearning, given the threat model, our experiments leverage the state-of-the-art client-level federated unlearning method introduced in \cite{Anisa23}.

\begin{itemize}[leftmargin=*]
    \item \textbf{Gradient ascent \cite{Anisa23}.} 
    To unlearn their data, a client optimizes the model parameters to maximize the loss function. 
\end{itemize}

In reinforcement unlearning, the objective is to forget an entire environment. Hence, our experiments exploit the two latest methods proposed in \cite{Ye25NDSS}: decremental reinforcement learning (RL)-based and poisoning-based methods.

\begin{itemize}[leftmargin=*]
    \item \textbf{Decremental RL.} It aims to diminish the agent's previous knowledge by letting it collect experience samples in the unlearning environment and then fine-tuning it with these samples to reduce its performance there.
    \item \textbf{Poisoning.} It encourages the agent to learn new, albeit incorrect, knowledge to remove the unlearning environment. It uses a poisoning approach to modify the unlearning environment and retrains the agent in this modified environment.
\end{itemize}

\vspace{1mm}
\noindent\textbf{Datasets and Model Architectures.}
\begin{itemize}[leftmargin=*]
    \item \textbf{CIFAR10} \cite{Krizhevsky14} includes $60,000$ images across $10$ classes, each containing $6,000$ images of vehicles and animals. The dimension of each image is $32\times 32$. 
    \item \textbf{MNIST} \cite{LeCun98} is a dataset of $70,000$ images of handwritten numerals spanning $10$ classes: $0-9$. Each class has $7000$ images and each image was resized to $32\times 32$.
    \item \textbf{SVHN} \cite{Netzer11NIPS} is a real-world street view house number dataset, consisting of $10$ classes, each representing a digit. It comprises over $600,000$ samples, each measuring $32\times 32$.
    \item \textbf{FaceScrub} \cite{Ng14} is a dataset of URLs for $100,000$ images of 530 individuals. We collected $91,712$ images of $526$ individuals. Each image was resized to $64\times 64$.
\end{itemize}

We implement two model architectures. Our custom architecture consists of four CNN blocks, followed by two fully-connected layers and a softmax function for classification. On the other hand, ResNet \cite{He15CVPR} includes multiple layers with residual blocks and uses projection shortcuts to match dimensions between layers.

\vspace{1mm}
\noindent\textbf{Evaluation Metrics.}
The assessment of unlearning outcomes typically involves three key approaches: membership inference \cite{Shokri17}, backdoor \cite{Pan22USENIX} techniques, and model accuracy \cite{Zhang22NIPS}. However, membership inference, especially in the context of sample-level unlearning, has limitations due to its reliance on discernible output differences between members and non-members of the training dataset \cite{Xu23}. In well-trained models, these differences are often too subtle, rendering membership inference ineffective for confirming unlearning. Therefore, our experiments focus on using backdoor techniques and model accuracy to measure the efficacy of unlearning.

For backdoor techniques, we utilize the approach designed for data ownership verification \cite{Li23TIFS}. This method intentionally induces misclassification in the model for the unlearning samples $\mathcal{D}_A$ that contain a specific `trigger'. After unlearning, these triggered samples should not be misclassified to the previously targeted class, indicating effective unlearning. For experimental consistency, we also integrate the same trigger into the victim data, $\mathcal{D}_V$, thereby making $\mathcal{D}_A$ and $\mathcal{D}_V$ effectively duplicate. We measure the effectiveness of unlearning using the \textit{Attack Success Rate}, a metric that evaluates how often the unlearned model $F_u$ incorrectly classifies the triggered data, $\mathcal{D}_A$ and $\mathcal{D}_V$, into a user-defined class. A lower attack success rate signifies a more successful unlearning outcome.
Note that these triggers are not intended for malicious purposes but used solely as tools to evaluate unlearning efficacy.

On the other hand, for model accuracy, the metrics \cite{Jia23NIPS} employed in our evaluation are outlined below.
\begin{itemize}[leftmargin=*]
    \item \textit{Model fidelity (MF).} This metric refers to the accuracy of unlearned model $F_u$ on the remaining data $\mathcal{D}_r$. 
    \item \textit{Testing accuracy (TA).} This metric assesses the generalization ability of the unlearned model $F_u$ on a test dataset $\mathcal{D}_t$ that excludes any samples from both $\mathcal{D}_u$ and $\mathcal{D}_r$. 
    \item \textit{Unlearning efficacy (UE).} We define this metric as $UE(F_u)=1-Acc_{\mathcal{D}_u}(F_u)$, where $Acc_{\mathcal{D}_u}(F_u)$ denotes the accuracy of $F_u$ on the unlearned data $\mathcal{D}_u$. 
    \item\textit{Unlearning impact (UI).} Unlearning impact is defined as $UI(F_u)=1-Acc_{\mathcal{D}_v}(F_u)$. $\mathcal{D}_v$ represents the victim data, either the same as or functionally equivalent to $\mathcal{D}_u$. 
\end{itemize}


\vspace{-4mm}
\subsection{Overall Results}
We assume that if the unlearning entities are identified as duplications by de-duplication techniques, the model owner will disregard the unlearning request to prevent potential adversarial exploitation. This assumption will be revisited later in the robustness study. 

\subsubsection{Machine Unlearning} 
Table \ref{tab:MLCIFARCustomBackdoor} shows the machine unlearning results on CIFAR10 using the custom model, with a focus on backdoor evaluation. It reveals that data duplication significantly impacts the unlearning process. Notably, the gold standard unlearning method, retraining from scratch, records a backdoor attack success rate of $99.6\%$ on the unlearned data $\mathcal{D}_A$. This high rate suggests that the retrained model retains knowledge of the trigger, indicating a failure to unlearning. Leveraging this failure, the adversary can assert that the model owner is dishonest, thereby achieving their goal. This outcome is due to the presence of a duplicate version of $\mathcal{D}_A$ in the training set, which means the retraining still exposes the model to the patterns and features of the unlearned data, including the trigger, resulting in a high attack success rate.

The Fisher forgetting method produces results similar to the retraining method, with a high backdoor attack success rate after unlearning, indicating a failure in unlearning even when de-duplication techniques are employed. As a result, the adversary can also achieve their objective when the model owner adopts the Fisher forgetting method. This failure arises because the Fisher forgetting approach, while designed to statistically reduce the influence of specific data points on the model's parameters, may not address the complex interdependencies between features in a neural network. When data duplication exists, the removal of influences associated with one set of data does not necessarily eliminate the residual effects from its duplicates, including triggers used in backdoor evaluations. Thus, despite the application of Fisher forgetting, the model continues to recognize and respond to these triggers, leading to a high attack success rate. 

By comparison, the relabeling method can achieve effective unlearning with a low backdoor attack success rate when de-duplication techniques are not applied. Thus, the relabeling method can effectively defend against the adversary’s attack. However, this comes at a significant cost, completely compromising the model’s performance by reducing test accuracy to approximately $11\%$. The introduction of de-duplication techniques complicates the outcome: the relabeling method fails to unlearn effectively when duplicates are completely identical, and severely damages model efficacy when duplicates are nearly identical. This occurs because the relabeling method relies on mislabeling the unlearned data to disrupt its influence on the model. In scenarios without de-duplication, mislabeling can sufficiently confuse the model about these data points. However, when de-duplication is active, if the duplicates are not exactly identical but close, the method struggles to differentiate between what needs to be unlearned and what does not. This confusion leads to the model learning incorrect information across similar but not identical duplicates, thereby degrading overall model performance.

\begin{table}[!ht]\scriptsize
\vspace{-2mm}
	\centering
	\caption{Machine unlearning results with backdoor evaluation on CIFAR10 using the custom model. Here, $FF$ stands for the Fisher forgetting method, $Rela.$ denotes the relabeling method, $CD$ refers to complete duplication, $ND$ represents near-duplication, and $DeDup$ denotes de-duplication.}
 \vspace{-0mm}
\begin{tabular} {ccccc} 
\toprule\rowcolor{gray!40}
  & \multicolumn{2}{c}{Accuracy} & \multicolumn{2}{c}{Attack Success Rate}\\
 \midrule
 & Train & Test & $\mathcal{D}_V$ & $\mathcal{D}_A$\\ \hline
Before unlearn  & $84.42$ & $68.62$ & $99.98$ & $99.98$ \\\hline \rowcolor{gray!40}
Retrain  & $85.35$ & $69.65$ & $99.96$ & $99.96$ \\ 
FF+CD w/o DeDup & $84.48$ & $68.64$ & $99.98$ & $99.98$ \\\rowcolor{gray!40} 
FF+CD with DeDup & $84.42$ & $68.62$ & $99.98$ & $99.98$  \\
FF+ND w/o DeDup & $86.5$ & $69.01$ & $100$ & $99.78$ \\\rowcolor{gray!40} 
FF+ND with DeDup & $86.5$ & $69.01$ & $99.95$ & $99.86$ \\ 
Rela.+CD w/o DeDup & $11.75$ & $11.77$ & $79.3$ & $1.14$ \\ \rowcolor{gray!40}
Rela.+CD with DeDup & $84.42$ & $68.62$ & $99.98$ & $99.98$ \\
Rela.+ND w/o DeDup & $11.87$ & $11.34$ & $99.98$ & $9.34$ \\\rowcolor{gray!40}
Rela.+ND with DeDup & $13.2$ & $12.88$ & $99.9$ & $6.68$ \\
\bottomrule
\end{tabular}
	\label{tab:MLCIFARCustomBackdoor}
 \vspace{-1mm}
\end{table}

Similar outcomes from the three approaches on the SVHN dataset can be observed in Table \ref{tab:MLSVHNCustomBackdoor}, demonstrating that the impact of data duplication is consistent across various datasets. This consistency underscores the pervasive nature of duplication issues in machine learning models.

\begin{table}[!ht]\scriptsize
\vspace{-2mm}
	\centering
	\caption{Machine unlearning results with backdoor evaluation on SVHN using the custom model.}
 \vspace{-0mm}
\begin{tabular} {ccccc} 
\toprule\rowcolor{gray!40}
  & \multicolumn{2}{c}{Accuracy} & \multicolumn{2}{c}{Attack Success Rate}\\
 \midrule
 & Train & Test & $\mathcal{D}_V$ & $\mathcal{D}_A$\\ \hline
Before unlearn  & $92.62$ & $87.04$ & $99.97$ & $99.99$ \\\hline \rowcolor{gray!40}
Retrain  & $92.62$ & $86.93$ & $99.96$ & $99.96$ \\ 
FF+CD w/o DeDup & $92.51$ & $86.98$ & $99.96$ & $99.99$ \\\rowcolor{gray!40} 
FF+CD with DeDup & $92.62$ & $87.04$ & $99.97$ & $99.99$ \\
FF+ND w/o DeDup & $94.95$ & $88.69$ & $99.85$ & $98.94$ \\\rowcolor{gray!40}
FF+ND with DeDup & $94.95$ & $88.69$ & $99.83$ & $98.92$ \\ 
Rela.+CD w/o DeDup & $24.5$ & $16.86$ & $100$ & $14.96$ \\ \rowcolor{gray!40}
Rela.+CD with DeDup & $92.62$ & $87.04$ & $99.97$ & $99.99$ \\
Rela.+ND w/o DeDup & $9.32$ & $8.97$ & $100$ & $0.85$ \\\rowcolor{gray!40}
Rela.+ND with DeDup & $10.2$ & $9.75$ & $100$ & $0.83$ \\
\bottomrule
\end{tabular}
	\label{tab:MLSVHNCustomBackdoor}
 \vspace{-1mm}
\end{table}

Table \ref{tab:MLCIFARCustom} presents the machine unlearning results with a model accuracy evaluation on CIFAR10 using the custom model. Notably, the traditional unlearning method of retraining from scratch demonstrates low unlearning efficacy at $27.8\%$, meaning the retrained model retains a high accuracy of $72.2\%$ on the unlearned data. This can be attributed to two factors. Firstly, the model may still be memorizing the features shared between the unlearned data and their duplicates. Secondly, the model's inherent generalizability might allow for high classification accuracy. However, given that the test accuracy of the retrained model is lower at $67.8\%$, it suggests that the model's high accuracy on the unlearned data is more likely due to residual memorization, caused by duplicated versions of the data still present in the training set.

\begin{table}[!ht]\scriptsize
\vspace{-1mm}
	\centering
	\caption{Machine unlearning results with model accuracy evaluation on CIFAR10 using the custom model.}
 \vspace{-0mm}
\begin{tabular} {ccccc} 
\toprule\rowcolor{gray!40}
  & Mode. fide. & Test accu. & Unle. effi. & Unle. impa.\\
 \midrule
Before unlearn  & $80.6$ & $68.1$ & $14.2$ & $13.0$ \\\hline \rowcolor{gray!40}
Retrain  & $81.5$ & $67.8$ & $27.8$ & $13.5$ \\
FF+CD w/o DeDup & $84.2$ & $67.5$ & $22.2$ & $5.2$ \\\rowcolor{gray!40}
FF+CD with DeDup & $80.6$ & $68.1$ & $14.2$ & $13.0$ \\
FF+ND w/o DeDup & $80.6$ & $68.2$ & $13.7$ & $13.0$ \\\rowcolor{gray!40}
FF+ND with DeDup & $80.6$ & $68.2$ & $13.7$ & $13.0$ \\
Rela.+CD w/o DeDup & $10.3$ & $9.8$ & 87.5 & $90.8$ \\ \rowcolor{gray!40}
Rela.+CD with DeDup & $80.6$ & $68.1$ & $14.2$ & $13.0$ \\
Rela.+ND w/o DeDup & $9.9$ & $10.4$ & $91.3$ & $91.5$ \\\rowcolor{gray!40}
Rela.+ND with DeDup & $12.7$ & $12.5$ & $88.5$ & $87.7$ \\
\bottomrule
\end{tabular}
	\label{tab:MLCIFARCustom}
 \vspace{-1mm}
\end{table}

For the Fisher forgetting method, a noteworthy phenomenon emerges: it consistently achieves very low unlearning efficacy, regardless of whether de-duplication is applied or whether complete or near duplication is present. This may be attributed to several factors. Firstly, the Fisher forgetting method relies on estimating the Fisher information matrix to compute the amount of forgetting required. However, in scenarios involving duplicated data, the Fisher information estimates may become unreliable due to the presence of redundant or correlated information. As a result, the method may struggle to accurately gauge the extent of forgetting necessary to remove duplicated instances from the model's memory.
Secondly, the Fisher forgetting method may inherently prioritize preserving information from unique data instances while inadvertently neglecting duplicated data. This bias towards unique instances could result in ineffective unlearning of duplicated data, leading to persistently low unlearning efficacy across various duplication scenarios.
Finally, the Fisher forgetting method's reliance on gradient-based optimization techniques may introduce challenges in scenarios with duplicated data. Gradient updates may inadvertently reinforce the model's memory of duplicated instances, thereby hindering the unlearning process. This contributes to the consistent low unlearning efficacy of the Fisher forgetting method.

For the relabeling method, in the absence of de-duplication techniques, it exhibits high unlearning efficacy. However, this comes at the expense of model fidelity and test accuracy in scenarios involving both complete and near duplication. The poor performance of the relabeling method can be attributed to several factors. Firstly, this method may struggle to effectively discern between original and duplicated data instances, leading to inaccurate updates to the model parameters during the unlearning process. Additionally, the inherent nature of relabeling may introduce noise into the unlearning process, particularly when applied to duplicated data, thereby undermining the model's ability to generalize to unseen samples. 

When de-duplication techniques are implemented, the relabeling method is not executed if they involve complete duplication, as complete duplication can be easily detected by the de-duplication techniques. However, in cases of near duplication, its outcomes mirror those achieved without applying de-duplication. This suggests that our proposed near-duplication methods enable duplicated data to circumvent detection, thereby still impacting the unlearning process.

\begin{table}[!ht]\scriptsize
\vspace{-1mm}
	\centering
	\caption{Machine unlearning results with model accuracy evaluation on SVHN using the custom model.}
 \vspace{-0mm}
\begin{tabular} {ccccc} 
\toprule\rowcolor{gray!40}
  & Mode. fide. & Test accu. & Unle. effi. & Unle. impa.\\
 \midrule
Before unlearn  & $93.8$ & $87.3$ & $15.2$ & $16.2$ \\\hline \rowcolor{gray!40}
Retrain  & $93$ & $86.5$ & $9.7$ & $6.3$ \\
Rela.+CD w/o DeDup & $13.9$ & $14.6$ & $86.7$ & $87.2$ \\ \rowcolor{gray!40} 
Rela.+CD with DeDup & $93.8$ & $87.3$ & $15.2$ & $16.2$ \\
Rela.+ND w/o DeDup & $7.6$ & $7.4$ & $92.2$ & $91.8$ \\ \rowcolor{gray!40} 
Rela.+ND with DeDup & $9.0$ & $8.5$ & $91.0$ & $92.2$ \\
FF+CD w/o DeDup & $94.3$ & $88.0$ & $5.3$ & $2.8$ \\ \rowcolor{gray!40}
FF+CD with DeDup & $93.8$ & $87.3$ & $15.2$ & $16.2$ \\
FF+ND w/o DeDup & $93.8$ & $87.4$ & $14.8$ & $6.3$ \\ \rowcolor{gray!40}
FF+ND with DeDup & $93.8$ & $87.4$ & $14.8$ & $6.3$ \\
\bottomrule
\end{tabular}
	\label{tab:MLSVHNCustom}
 \vspace{-1mm}
\end{table}

Table \ref{tab:MLSVHNCustom} shows the machine unlearning outcomes with a model accuracy evaluation on SVHN using the custom model. These results exhibit a notable consistency with those observed on CIFAR10, as presented in Table \ref{tab:MLCIFARCustom}. This consistency underscores the persistence of both the impact of duplicated data and the effectiveness of our near-duplication method across different datasets.

\vspace{2mm}
\noindent\textbf{Summary.} In standard machine unlearning, the adversary easily achieves their goal of challenging the model owner’s success in unlearning under both the retraining-from-scratch and Fisher forgetting methods. The relabeling method can defend against the adversary’s attack by significantly reducing the backdoor attack success rate and increasing the unlearning efficacy, but it does so at a substantial cost of completely degrading the model’s performance.

\subsubsection{Federated Unlearning} 
In federated unlearning, the unlearned data $\mathcal{D}_u$ is redefined to represent the private data of the unlearned client. The remaining data $\mathcal{D}_r$ refers to the private data of all other clients. The test dataset $\mathcal{D}_t$ denotes the local dataset held by the server and utilized for testing the accuracy of the global model. 

Table \ref{tab:FedMNISTCustom} shows the federated unlearning results on MNIST with the custom model. We observe that before unlearning, the model demonstrates high fidelity ($98.44\%$) and testing accuracy ($98.18\%$), coupled with low unlearning efficacy ($1.4\%$). This outcome arises from the model's effective training without any unlearning procedures. Then, employing the retraining-from-scratch approach to unlearn the global model allows for the preservation of its performance. This occurs because the retraining approach simply discards those duplicates and utilizes the remaining clients to retrain a global model. Furthermore, since the adversary client's data has never been incorporated into the training process, the retrained global model exhibits a deterioration in performance when evaluated on the adversary client's data, consequently yielding a high unlearning efficacy. 
Thus, the retraining-from-scratch approach effectively counters the adversary's objective. However, this approach is often impractical in real-world scenarios, particularly when the number of clients is large, due to its significant computational and resource demands.

\begin{table}[!ht]\scriptsize
\vspace{-2mm}
	\centering
 \vspace{-0mm}
	\caption{Federated unlearning results on MNIST with the custom model. Here, $GA$ refers to the gradient ascent method.}
\begin{tabular} {cccc} 
\toprule\rowcolor{gray!40}
  & Model fidelity & Test accuracy & Unlearn efficacy \\
 \midrule
Before unlearn  & $98.44$ & $98.18$ & $1.4$  \\\hline \rowcolor{gray!40}
Retrain & $96.45$ & $97.82$ & $85.16$   \\
GA+CD w/o DeDup & $17.89$ & $18.46$ & $96.32$  \\\rowcolor{gray!40}
GA+CD with DeDup & $98.03$ & $98.11$ & $1.7$  \\
GA+ND w/o DeDup & $39.11$ & $39.38$ & $92.38$ \\\rowcolor{gray!40}
GA+ND with DeDup & $43.33$ & $43.1$ & $91.48$ \\
\bottomrule
\end{tabular}
	\label{tab:FedMNISTCustom}
 \vspace{-2mm}
\end{table}

When the model owner employs the gradient ascent unlearning approach, a significant decline in model fidelity and testing accuracy is observed, accompanied by a notable increase in unlearning efficacy. A closer examination of the third to sixth rows reveals several interesting phenomena. In the third row, we observe that without applying any de-duplication technique, unlearning the completely duplicated local models can demolish the global model, resulting in a very low model fidelity ($17.89\%$) and testing accuracy ($18.46\%$), while achieving a remarkably high unlearning efficacy ($96.32\%$). This occurs because unlearning these complete duplicates is equivalent to unlearning the entire global model, given that the global model was aggregated from those local models. However, such complete duplication is easily identified by de-duplication techniques, as demonstrated in the fourth row. In cases where the server declines to unlearn duplicates, the model's performance can be preserved. 

With the introduction of our near-duplication method, the results become even more intriguing. In situations where no de-duplication technique is applied (the fifth row), unlearning near-duplicates leads to higher model fidelity and testing accuracy with a lower unlearning efficacy compared to the complete duplication scenario (the third row). This is because near-duplicates exhibit minor differences compared to complete duplicates, which can prevent the model from destruction. 
However, these differences can also enable near-duplicates to evade detection by de-duplication techniques (the sixth row), resulting in similar outcomes to the scenario without de-duplication (the fifth row). This implies that carefully designed near-duplicates can evade existing de-duplication techniques. 

The above results highlight the adversary's successful achievement of their objective: significantly reducing the global model's utility. While de-duplication techniques can effectively detect complete duplicates, they are easily circumvented by our crafted near-duplicates.


\begin{table}[!ht]\scriptsize
\vspace{-2mm}
	\centering
	\caption{Federated unlearning results on CIFAR10 with the custom model.}
 \vspace{-0mm}
\begin{tabular} {cccc} 
\toprule\rowcolor{gray!40}
  & Model fidelity & Test accuracy & Unlearn efficacy \\
 \midrule
Before unlearn  & $65.75$ & $62.77$ &  $47.44$  \\\hline \rowcolor{gray!40}
Retrain & $55.62$ & $58.75$ &  $55.28$  \\
GA+CD w/o DeDup & $9.96$ & $10.01$ & $97.83$  \\\rowcolor{gray!40}
GA+CD with DeDup & $62.04$ & $61.68$ & $32.36$  \\
GA+ND w/o DeDup & $10.08$ & $9.98$ & $98.19$ \\\rowcolor{gray!40}
GA+ND with DeDup & $10.15$ & $10.24$ & $97.81$ \\
\bottomrule
\end{tabular}
	\label{tab:FedCIFARCustom}
 \vspace{-2mm}
\end{table}

Table \ref{tab:FedCIFARCustom} presents the federated unlearning outcomes on CIFAR10 with the custom model. Generally, the results align with those observed on MNIST, as depicted in Table \ref{tab:FedMNISTCustom}, indicating that unlearning duplicates also results in negative impacts on CIFAR10. The primary difference between the two tables lies in the third and fifth rows: on CIFAR10, unlearning near-duplicates does not yield much difference from unlearning complete duplicates, whereas this contrast is not observed on MNIST. This dissimilarity can be attributed to the complexity disparity between the two datasets. CIFAR10 is inherently more complex than MNIST, including the diversity of objects and background clutter. As a result, the local models computed during training on CIFAR10 are likely to exhibit greater complexity compared to those on MNIST.

When unlearning near-duplicates, the slight differences between local models may be more significant in CIFAR10 due to its higher complexity. Thus, even near-duplicates can have a substantial impact on the model's performance, leading to outcomes similar to those observed with complete duplicates.
On the other hand, MNIST, being a simpler dataset, may be more resilient to minor differences in local models. Thus, the impact of near-duplicates on model performance may be less pronounced on MNIST compared to CIFAR10.

\begin{table}[!ht]\scriptsize
\vspace{-2mm}
	\centering
	\caption{Federated unlearning results on FaceScrub with the custom model.}
 \vspace{-0mm}
\begin{tabular} {cccc} 
\toprule\rowcolor{gray!40}
  & Model fidelity & Test accuracy & Unlearn efficacy \\
 \midrule
Before unlearn  & $89.31$ & $88.87$ &  $10.88$  \\\hline \rowcolor{gray!40}
Retrain & $85.59$ & $86.78$ &  $82.82$  \\
GA+CD w/o DeDup & $15.62$ & $17.42$ & $95.84$  \\\rowcolor{gray!40}
GA+CD with DeDup & $88.52$ & $87.89$ & $11.24$  \\
GA+ND w/o DeDup & $36.57$ & $35.73$ & $94.68$ \\\rowcolor{gray!40}
GA+ND with DeDup & $41.41$ & $40.25$ & $88.10$ \\
\bottomrule
\end{tabular}
	\label{tab:FedFaceCustom}
 \vspace{-1mm}
\end{table}

Table \ref{tab:FedFaceCustom} displays the federated unlearning outcomes on FaceScrub with the custom model. Remarkably, we observe a consistent trend akin to the results obtained on MNIST and CIFAR10 datasets. This underscores the adaptability and robustness of our proposed near-duplication methods across diverse datasets. Moreover, it highlights the pervasive impact of duplicates on the unlearning process, irrespective of the dataset characteristics.

\vspace{2mm}
\noindent\textbf{Summary.} In federated unlearning, the adversary successfully achieves their goal of severely compromising the global model under the gradient ascent approach. In contrast, their objective is less attainable under the retraining-from-scratch approach. However, retraining poses practical challenges in real-world scenarios, particularly when the number of clients is large.

\subsection{Adaptability Study}
\subsubsection{Alternative Model Architectures}
We conducted conventional machine unlearning and federated unlearning using ResNet, while employing Deep Deterministic Policy Gradient (DDPG) for reinforcement unlearning. DDPG comprises an actor network, responsible for learning a deterministic policy, and a critic network, tasked with evaluating the quality of actions selected by the actor \cite{DDPG}.

\vspace{2mm}
\noindent\textbf{Machine Unlearning.} The results of machine unlearning on CIFAR10 and SVHN using ResNet are presented in Tables \ref{tab:MLCIFARResNet} and \ref{tab:MLSVHNResNet}, respectively. The trends observed are consistent with those obtained using the custom model. However, there are several noteworthy differences that warrant attention.

\begin{table}[!ht]\scriptsize
\vspace{-1mm}
	\centering
	\caption{Machine unlearning results on CIFAR10 with ResNet.}
 \vspace{-0mm}
\begin{tabular} {ccccc} 
\toprule\rowcolor{gray!40}
  & Mode. fide. & Test accu. & Unle. effi. & Unle. Impa. \\
 \midrule
Before unlearn  & $96.6$ & $80.8$ & $0$ & $0$ \\\hline \rowcolor{gray!40}
Retrain  & $96.6$ & $80.7$ & $12.8$ & $0$ \\
FF+CD w/o DeDup & $96.5$ & $80.1$ & $10.8$ & $0$\\\rowcolor{gray!40}
FF+CD with DeDup & $96.6$ & $80.8$ & $0$ & $0$\\
FF+ND w/o DeDup & $96.5$ & $80.0$ & $0$ & $0$ \\\rowcolor{gray!40}
FF+ND with DeDup & $96.5$ & $80.0$ & $0$ & $0$ \\
Rela.+CD w/o DeDup & $11.5$ & $11.2$ & $90.0$ & $89.8$ \\\rowcolor{gray!40}
Rela.+CD with DeDup & $96.6$ & $80.8$ & $0$ & $0$ \\
Rela.+ND w/o DeDup & $9.0$ & $9.0$ & $93.3$ & $92.5$ \\\rowcolor{gray!40}
Rela.+ND with DeDup & $12.7$ & $12.5$ & $88.5$ & $87.7$ \\
\bottomrule
\end{tabular}
	\label{tab:MLCIFARResNet}
 \vspace{-2mm}
\end{table}

\begin{table}[!ht]\scriptsize
\vspace{-2mm}
	\centering
	\caption{Machine unlearning results on SVHN with ResNet.}
 \vspace{-0mm}
\begin{tabular} {ccccc} 
\toprule\rowcolor{gray!40}
  & Mode. fide. & Test accu. & Unle. effi. & Unle. Impa. \\
 \midrule
Before unlearn  & $100$ & $94.5$ & $0$ & $0$ \\\hline \rowcolor{gray!40}
Retrain  & $100$ & $94.4$ & $0$ & $0$ \\
FF+CD w/o DeDup & $100$ & $93.8$ & $0$ & $0$\\\rowcolor{gray!40}
FF+CD with DeDup & $100$ & $94.5$ & $0$ & $0$\\
FF+ND w/o DeDup & $100$ & $94.4$ & $0$ & $0$ \\\rowcolor{gray!40}
FF+ND with DeDup & $100$ & $94.4$ & $0$ & $0$ \\
Rela.+CD w/o DeDup & $7.4$ & $7.3$ & $93.2$ & $93.8$ \\\rowcolor{gray!40}
Rela.+CD with DeDup & $100$ & $94.5$ & $0$ & $0$ \\
Rela.+ND w/o DeDup & $9.7$ & $9.5$ & $91.0$ & $90.8$ \\\rowcolor{gray!40}
Rela.+ND with DeDup & $9.1$ & $8.8$ & $90.0$ & $90.0$ \\
\bottomrule
\end{tabular}
	\label{tab:MLSVHNResNet}
 \vspace{-1mm}
\end{table}

As depicted in Table \ref{tab:MLSVHNResNet}, the retraining-from-scratch approach attains an unlearning efficacy as low as $0$, indicating a classification accuracy of $100\%$ on the unlearned data. This failure of the gold standard approach signifies its inability to forget information from the unlearned data in this scenario. This failure could be attributed to two factors.
Firstly, ResNet models are renowned for their high generalization ability, allowing them to effectively learn complex patterns from data. However, this same capability may hinder their ability to unlearn information, particularly in scenarios involving duplicated data. 
Furthermore, the high capacity and expressiveness of ResNet models may result in the model memorizing the duplicated instances, rather than learning to generalize and discard them during unlearning. This phenomenon can lead to poor unlearning efficacy, as observed in the results.

In both Tables \ref{tab:MLCIFARResNet} and \ref{tab:MLSVHNResNet}, a notable observation is the failure of the Fisher forgetting method to unlearn duplicated data, regardless of the application of de-duplication or the presence of complete or near duplication. This phenomenon suggests inherent limitations in the Fisher forgetting method when confronted with duplicated data.
One possible explanation for this failure is the method's reliance on gradient-based optimization techniques. The Fisher forgetting method estimates the Fisher information matrix to compute the amount of forgetting necessary for unlearning. However, in scenarios with duplicated data, gradients may become unstable or misleading, making it challenging for the method to accurately estimate the required forgetting.

As a result, the inability of both the retraining-from-scratch and Fisher forgetting approaches to achieve effective unlearning demonstrates the adversary’s success in challenging the model owner’s claim of unlearning efficacy.

\vspace{2mm}
\noindent\textbf{Federated unlearning.} The results of federated unlearning on CIFAR10 using ResNet are shown in Table \ref{tab:FedCIFARResNet}, mirroring a similar trend observed in the results obtained with the custom model, as depicted in Table \ref{tab:FedCIFARCustom}. 
These findings highlight the adversary’s success in achieving their objective of significantly degrading the global model's performance.

A notable difference between the two tables lies in the model fidelity and test accuracy achieved using ResNet, which generally outperforms the custom model. This disparity can be attributed to the high capability of ResNet in capturing complex patterns and features inherent in the CIFAR10 dataset. However, despite its superior performance, ResNet is still susceptible to the adverse effects of duplicated data.

\begin{table}[!ht]\scriptsize
\vspace{-3mm}
	\centering
	\caption{Federated unlearning on CIFAR10 with ResNet.}
 \vspace{-0mm}
\begin{tabular} {cccc} 
\toprule\rowcolor{gray!40}
  & Model fidelity & Test accuracy & Unlearn efficacy \\
 \midrule
Before unlearn  & $83.76$ & $83.56$ &  $15.42$  \\\hline \rowcolor{gray!40}
Retrain & $74.53$ & $75.42$ &  $24.21$  \\
GA+CD w/o DeDup & $22.41$ & $21.54$ & $96.66$  \\\rowcolor{gray!40}
GA+CD with DeDup & $80.22$ & $82.14$ & $17.5$  \\
GA+ND w/o DeDup & $46.52$ & $46.45$ & $97.62$ \\\rowcolor{gray!40}
GA+ND with DeDup & $46.64$ & $45.21$ & $90.78$ \\
\bottomrule
\end{tabular}
	\label{tab:FedCIFARResNet}
 \vspace{-2mm}
\end{table}

\vspace{-2mm}
\subsubsection{Varying Numbers of Unlearned Data}
For machine unlearning, the previous experimental results were obtained using $600$ unlearned data. We extend this number to $2,000$ and evaluate the effects of this increased quantity. 
In federated unlearning, the emphasis is on unlearning clients, so we raise the number of unlearned clients from $1$ to $2$. 

\vspace{2mm}
\noindent\textbf{Machine Unlearning.} The results of unlearning $2,000$ data on CIFAR10 and SVHN are displayed in Tables \ref{tab:MLCIFARCustomUnlearn2000} and \ref{tab:MLSVHNCustomUnlearn2000}, respectively. Interestingly, these results closely mirror those obtained when unlearning $600$ data. This suggests that the impact of duplicated data remains consistent regardless of the quantity of unlearned data, with the only expected difference being a decrease in model fidelity and test accuracy due to the additional unlearned data.
This observation shows the adversary's success in achieving their objective and underscores the robust and persistent nature of the impact of duplicated data on the unlearning process. Despite varying the quantity of unlearned data, the presence of duplicated data continues to exert a significant influence on model performance.

\begin{table}[!ht]\scriptsize
\vspace{-3mm}
	\centering
	\caption{Machine unlearning results on CIFAR10 with $2,000$ unlearned data.}
 \vspace{-0mm}
\begin{tabular} {ccccc} 
\toprule\rowcolor{gray!40}
  & Mode. fide. & Test accu. & Unle. effi. & Unle. impa.\\
 \midrule
Before unlearn  & $80.9$ & $67.87$ & $18.4$ & $12.4$ \\\hline \rowcolor{gray!40}
Retrain  & $86.5$ & $68.0$ & $32.2$ & $10.5$ \\ 
FF+ND w/o DeDup & $81.0$ & $68.0$ & $18.4$ & $12.3$ \\ \rowcolor{gray!40}
Rela.+ND w/o DeDup & $17.9$ & $17.1$ & $83.4$ & $81.1$ \\

\bottomrule
\end{tabular}
	\label{tab:MLCIFARCustomUnlearn2000}
 \vspace{-2mm}
\end{table}

\begin{table}[!ht]\scriptsize
\vspace{-3mm}
	\centering
	\caption{Machine unlearning results on SVHN with $2,000$ unlearned data.}
 \vspace{-0mm}
\begin{tabular} {ccccc} 
\toprule\rowcolor{gray!40}
  & Mode. fide. & Test accu. & Unle. effi. & Unle. impa.\\
 \midrule
Before unlearn  & $96.2$ & $87.4$ & $2.8$ & $2.2$ \\\hline \rowcolor{gray!40}
Retrain  & $95.7$ & $87.2$ & $9.1$ & $4.3$ \\
FF+ND w/o DeDup & $96.2$ & $87.4$ & $2.8$ & $2.2$ \\ \rowcolor{gray!40}
Rela.+ND w/o DeDup & $13.6$ & $14.5$ & $85.4$ & $86.2$ \\
\bottomrule
\end{tabular}
	\label{tab:MLSVHNCustomUnlearn2000}
 \vspace{-2mm}
\end{table}

The consistent effect of duplicated data across different quantities of unlearned data can be attributed to several factors. Firstly, duplicated data introduces redundancy into the training process, which can interfere with the model's ability to generalize effectively. Additionally, duplicated data may lead to overfitting, where the model learns to memorize the duplicated instances rather than generalize from them.

\vspace{2mm}
\noindent\textbf{Federated Unlearning.} The results of unlearning $2$ clients on MNIST, as shown in Table \ref{tab:FedMNISTCustom2Clients}, align with the outcomes observed when unlearning $1$ client, indicating the adversary's success in achieving their objective. This observation underscores that unlearning $2$ clients does not incur additional damage to the global model compared to unlearning $1$ client. This can be attributed to the inherent adaptability of the federated learning framework. In this framework, unlearning a client is akin to removing that client from the global model training process. This adaptability allows the FL system to effectively cope with the unlearning of multiple clients without experiencing proportionate harm to the global model.

\begin{table}[!ht]\scriptsize
\vspace{-3mm}
	\centering
	\caption{Federated unlearning results on MNIST with $2$ unlearned clients.}
 \vspace{-0mm}
\begin{tabular} {cccc} 
\toprule\rowcolor{gray!40}
  & Model fidelity & Test accuracy & Unlearn efficacy \\
 \midrule
Before unlearn  & $98.44$ & $98.16$ & $1.4$  \\\hline \rowcolor{gray!40}
Retrain & $96.57$ & $97.70$ & $85.67$   \\
GA+CD w/o DeDup & $16.52$ & $16.25$ & $96.73$  \\\rowcolor{gray!40}
GA+CD with DeDup & $97.86$ & $98.48$ & $2.12$  \\
GA+ND w/o DeDup & $39.11$ & $39.14$ & $91.89$ \\\rowcolor{gray!40}
GA+ND with DeDup & $44.43$ & $42.96$ & $92.05$ \\
\bottomrule
\end{tabular}
	\label{tab:FedMNISTCustom2Clients}
 \vspace{-2mm}
\end{table}

\vspace{2mm}
\noindent\textbf{Summary.} Complex models, such as ResNet, exhibit greater vulnerability to duplicates compared to our customized models in machine unlearning. However, in federated unlearning, complex models demonstrate superior performance over customized models when duplicates are present. Furthermore, the number of unlearned data points or clients has minimal impact on unlearning outcomes in both machine unlearning and federated unlearning with the presence of duplicates.

\subsection{Robustness Study}
In the previous results, we assumed that unlearning would be denied when duplications are detected. Here, we conduct a detailed study of the robustness of our methods by proposing two countermeasures. Note that these countermeasures are described within machine unlearning but can be easily extended to federated and reinforcement unlearning by substituting `data' with `model updates' and `environments', respectively.

\begin{itemize}[leftmargin=*]
    \item For detected duplicates, the model owner unlearns both the requested data and their duplicates. Formally, suppose the adversary requests to unlearn their data $\mathcal{D}_A$, and a portion of this data is detected by the model owner as duplicates. Let the detected duplicates be $\mathcal{D}'_A \subseteq \mathcal{D}_A$, with the corresponding original data in the training set denoted as $\mathcal{D}'_V$. The model owner unlearns both $\mathcal{D}_A$ and $\mathcal{D}'_V$ to mitigate potential adversarial activities, such as verification challenges.
    \item For detected duplicates $\mathcal{D}'_A$, the model owner removes the corresponding original data $\mathcal{D}'_V$ and unlearns the requested data $\mathcal{D}_A$ to mitigate the risk of potential poisoned data.
\end{itemize}

In addition to evaluating the adopted de-duplication strategies, we also assess the most advanced scenario by assuming a perfect model owner capable of fully detecting all duplicates. Furthermore, to simplify the explanation of evaluation results, we introduce a new metric: \emph{Acc. on $\mathcal{D}_u$}, which represents the model's accuracy on the unlearned data $\mathcal{D}_u$.

\vspace{2mm}
\noindent\textbf{Machine Unlearning.} Tables \ref{tab:MLCIFARDefense1-50} and \ref{tab:MLCIFARDefense1-100} present the results of applying the first countermeasure: unlearning both the requested data and their detected duplicates under the adopted and perfect de-duplication strategies, respectively. By comparing the two tables, it is evident that the Fisher forgetting and relabeling methods, even when fine-tuned on the remaining data, yield very similar results. This indicates that identifying only a portion of the duplicates and identifying all duplicates produce similar results for these unlearning methods. This similarity arises because both methods fundamentally alter the model's internal representations to diminish the influence of the unlearned data. The Fisher forgetting method statistically reduces the impact of specific data points on the model's parameters, while the relabeling method disrupts the model's learned associations by introducing intentional mislabeling. In scenarios where duplicates are either partially or fully detected, these methods behave similarly in terms of their capacity to degrade the model's reliance on duplicated data. Specifically, both the Fisher forgetting method and the relabeling method with fine-tuning achieve very high accuracy on the unlearned data, exceeding not only the test accuracy but also the fidelity of the remaining training data. This allows the adversary to challenge the model owner's success in unlearning during verification. On the other hand, the relabeling method without fine-tuning achieves low accuracy on the unlearned data, supporting unlearning verification. However, this comes at the significant cost of severely compromising the model’s utility, as evidenced by its extremely low test accuracy.

The key difference between the two tables lies in the performance of the retraining-from-scratch approach. Under the adopted de-duplication strategy (Table \ref{tab:MLCIFARDefense1-50}), approximately $50\%$ of the duplicates are detected. In this scenario, the retraining process includes the remaining $50\%$ of undetected duplicates, resulting in a high accuracy on the unlearned data. As this accuracy surpasses the test accuracy, the adversary can still challenge the success of unlearning. Conversely, under the perfect de-duplication strategy (Table \ref{tab:MLCIFARDefense1-100}), all duplicates are accurately detected. As a result, the retraining process excludes all duplicates and their original counterparts, leading to a low accuracy on the unlearned data comparable to the test accuracy. In this ideal case, the retraining approach can withstand the adversary's verification challenge while preserving the model's utility. However, the retraining approach is rarely used in practice due to its high computational overhead.

\begin{table}[!ht]\scriptsize
\vspace{-2mm}
	\centering
	\caption{Machine unlearning results on CIFAR10 using the first countermeasure: unlearning both the requested data and their detected duplicates, under the adopted de-duplication.}
 \vspace{-0mm}
\begin{tabular} {ccccc} 
\toprule\rowcolor{gray!40}
  & Mode. fide. & Test accu. & Unle. effi. & Acc. on $\mathcal{D}_u$\\
 \midrule
Before unlearn  & $77.2$ & $69.1$ & $15.9$ & $84.1$ \\\hline \rowcolor{gray!40}
Retrain  & $79.5$ & $68.8$ & $24.1$ & $75.9$ \\
FF & $77.1$ & $69.0$ & $16.1$ & $83.9$ \\\rowcolor{gray!40}
Rela. & $13.6$ & $13.2$ & $86.2$ & $13.8$ \\
Rela.+FT & $73.2$ & $65.6$ & $24.4$ & $75.6$ \\
\bottomrule
\end{tabular}
	\label{tab:MLCIFARDefense1-50}
 \vspace{-2mm}
\end{table}

\begin{table}[!ht]\scriptsize
\vspace{-2mm}
	\centering
	\caption{Machine unlearning results on CIFAR10 using the first countermeasure under the perfect de-duplication.}
 \vspace{-0mm}
\begin{tabular} {ccccc} 
\toprule\rowcolor{gray!40}
  & Mode. fide. & Test accu. & Unle. effi. & Acc. on $\mathcal{D}_u$\\
 \midrule
Before unlearn  & $77.2$ & $69.1$ & $15.9$ & $84.1$ \\\hline \rowcolor{gray!40}
Retrain  & $80.3$ & $67.7$ & $30.5$ & $69.5$ \\
FF & $76.9$ & $69.0$ & $16.1$ & $83.9$ \\\rowcolor{gray!40}
Rela. & $11.3$ & $11.4$ & $87.5$ & $12.5$ \\
Rela.+FT & $73.1$ & $65.6$ & $24.4$ & $75.6$ \\
\bottomrule
\end{tabular}
	\label{tab:MLCIFARDefense1-100}
 \vspace{-1mm}
\end{table}

Tables \ref{tab:MLCIFARDefense2-50} and \ref{tab:MLCIFARDefense2-100} present the defense results using the second countermeasure: removing the corresponding original data of the detected duplicates and unlearning the requested data, under the adopted and perfect de-duplication strategies, respectively. The results of all three methods - retraining, Fisher forgetting, and relabeling (including its fine-tuning variant) - exhibit a strong resemblance to the corresponding outcomes from the first countermeasure. This is because both countermeasures aim to address the influence of duplicates and their associated originals during the unlearning process. By either unlearning the duplicates and their originals (first countermeasure) or removing the originals while unlearning the requested data (second countermeasure), the processes effectively target overlapping information embedded within these datasets. As a result, the model undergoes a similar degree of influence reduction in both scenarios, leading to comparable effects on model performance. 

\begin{table}[!ht]\scriptsize
\vspace{-2mm}
	\centering
	\caption{Machine unlearning results on CIFAR10 using the second countermeasure: removing the corresponding original data of the detected duplicates and unlearning the requested data, under the adopted de-duplication.}
 \vspace{-0mm}
\begin{tabular} {ccccc} 
\toprule\rowcolor{gray!40}
  & Mode. fide. & Test accu. & Unle. effi. & Acc. on $\mathcal{D}_u$\\
 \midrule
Before unlearn  & $77.2$ & $69.1$ & $15.9$ & $84.1$ \\\hline \rowcolor{gray!40}
Retrain  & $79.5$ & $68.8$ & $24.1$ & $75.9$ \\
FF & $77.1$ & $69.0$ & $16.1$ & $83.9$ \\\rowcolor{gray!40}
Rela. & $13.9$ & $13.7$ & $86.4$ & $13.6$ \\
Rela.+FT & $72.9$ & $65.4$ & $24.7$ & $75.3$ \\
\bottomrule
\end{tabular}
	\label{tab:MLCIFARDefense2-50}
 \vspace{-2mm}
\end{table}

\begin{table}[!ht]\scriptsize
\vspace{-2mm}
	\centering
	\caption{Machine unlearning results on CIFAR10 using the second countermeasure under the perfect de-duplication.}
 \vspace{-0mm}
\begin{tabular} {ccccc} 
\toprule\rowcolor{gray!40}
  & Mode. fide. & Test accu. & Unle. effi. & Acc. on $\mathcal{D}_u$\\
 \midrule
Before unlearn  & $77.0$ & $69.1$ & $16.0$ & $84.0$ \\\hline \rowcolor{gray!40}
Retrain  & $80.3$ & $67.8$ & $30.6$ & $69.4$ \\
FF & $76.9$ & $69.0$ & $16.1$ & $83.9$ \\\rowcolor{gray!40}
Rela. & $13.8$ & $13.7$ & $86.4$ & $13.6$ \\
Rela.+FT & $77.0$ & $68.3$ & $22.4$ & $77.6$ \\
\bottomrule
\end{tabular}
	\label{tab:MLCIFARDefense2-100}
 \vspace{-1mm}
\end{table}

\vspace{2mm}
\noindent\textbf{Federated Unlearning.} Tables \ref{tab:FLCIFARDefense1-50} and \ref{tab:FLCIFARDefense1-100} present the results of applying the first countermeasure: unlearning both the requested model updates and their detected duplicates under the adopted and perfect de-duplication strategies, respectively. The results reveal that under both de-duplication strategies, the gradient ascent method severely degrades the global model's performance after unlearning, despite achieving high unlearning efficacy. This outcome occurs because unlearning both the requested model updates and their detected duplicates  erases a significant portion of the global model’s learned knowledge, rendering the model unable to retain its utility.

An interesting phenomenon is that the gradient ascent method performs even worse under the perfect de-duplication strategy (Table \ref{tab:FLCIFARDefense1-100}) compared to the adopted de-duplication strategy (Table \ref{tab:FLCIFARDefense1-50}). Specifically, the more duplicates detected and unlearned, the worse the unlearning outcomes. This can be attributed to the fact that the perfect de-duplication strategy identifies and removes all duplicate updates, including their original counterparts. Consequently, the unlearning process targets a larger portion of the model's training history, leading to over-unlearning. By aggressively reversing the influence of these updates, the gradient ascent method essentially forgets the foundational knowledge embedded in the global model, exacerbating its degradation.

In contrast, the retraining-from-scratch method successfully preserves the utility of the global model under both de-duplication strategies. This is attributable to the inherent characteristics of the retraining approach, which automatically removes the requested model updates, discards the existing global model, and constructs a new global model entirely from scratch. This process ensures that the influence of duplicated updates is eliminated without compromising the overall functionality of the retrained model.

\begin{table}[!ht]\scriptsize
\vspace{-2mm}
	\centering
	\caption{Federated unlearning results on CIFAR10 using the first countermeasure under the adopted de-duplication.}
 \vspace{-0mm}
\begin{tabular} {ccccc} 
\toprule\rowcolor{gray!40}
  & Mode. fide. & Test accu. & Unle. effi. & Acc. on $\mathcal{D}_u$\\
 \midrule
Before unlearn  & $65.75$ & $62.77$ & $47.44$ & $52.56$ \\\hline \rowcolor{gray!40}
Retrain  & $59.12$ & $60.56$ & $53.45$ & $46.55$ \\ 
GA & $26.03$ & $34.89$ & $93.85$ & $6.15$ \\
\bottomrule
\end{tabular}
	\label{tab:FLCIFARDefense1-50}
 \vspace{-3mm}
\end{table}

\begin{table}[!ht]\scriptsize
\vspace{-2mm}
	\centering
	\caption{Federated unlearning results on CIFAR10 using the first countermeasure under the perfect de-duplication.}
 \vspace{-0mm}
\begin{tabular} {ccccc} 
\toprule\rowcolor{gray!40}
  & Mode. fide. & Test accu. & Unle. effi. & Acc. on $\mathcal{D}_u$\\
 \midrule
Before unlearn  & $65.75$ & $62.77$ & $47.44$ & $52.56$ \\\hline \rowcolor{gray!40}
Retrain  & $58.45$ & $60.04$ & $53.98$ & $46.02$ \\ 
GA & $22.24$ & $32.23$ & $96.21$ & $3.79$ \\
\bottomrule
\end{tabular}
	\label{tab:FLCIFARDefense1-100}
 \vspace{-2mm}
\end{table}

Similar results are observed when employing the second countermeasure: removing the corresponding original model updates of the detected duplicates and unlearning the requested updates (Tables \ref{tab:FLCIFARDefense2-50} and \ref{tab:FLCIFARDefense2-100}). This is due to the fact that both countermeasures result in the removal of similar portions of the model updates. The second countermeasure additionally ensures that the original updates associated with detected duplicates are also erased, which mirrors the cumulative effect of unlearning applied in the first countermeasure.

\begin{table}[!ht]\scriptsize
\vspace{-2mm}
	\centering
	\caption{Federated unlearning results on CIFAR10 using the second countermeasure under the adopted de-duplication.}
 \vspace{-0mm}
\begin{tabular} {ccccc} 
\toprule\rowcolor{gray!40}
  & Mode. fide. & Test accu. & Unle. effi. & Acc. on $\mathcal{D}_u$\\
 \midrule
Before unlearn  & $65.75$ & $62.77$ & $47.44$ & $52.56$ \\\hline \rowcolor{gray!40}
Retrain  & $61.24$ & $60.02$ & $54.89$ & $45.11$ \\ 
GA & $23.98$ & $33.02$ & $91.98$ & $8.02$ \\
\bottomrule
\end{tabular}
	\label{tab:FLCIFARDefense2-50}
 \vspace{-3mm}
\end{table}

\begin{table}[!ht]\scriptsize
\vspace{-3mm}
	\centering
	\caption{Federated unlearning results on CIFAR10 using the second countermeasure under the perfect de-duplication.}
 \vspace{-0mm}
\begin{tabular} {ccccc} 
\toprule\rowcolor{gray!40}
  & Mode. fide. & Test accu. & Unle. effi. & Acc. on $\mathcal{D}_u$\\
 \midrule
Before unlearn  & $65.75$ & $62.77$ &  $47.44$ & $52.56$ \\\hline \rowcolor{gray!40}
Retrain  & $60.58$ & $59.65$ & $55.44$ & $44.56$ \\
GA & $21.87$ & $30.41$ & $94.52$ & $5.48$ \\
\bottomrule
\end{tabular}
	\label{tab:FLCIFARDefense2-100}
 \vspace{-2mm}
\end{table}

\vspace{2mm}
\noindent\textbf{Summary.} The adversary can achieve their objectives in both machine unlearning and federated unlearning in most scenarios, even when defense mechanisms and de-duplication techniques are applied. The retraining approach effectively resists the adversary in machine unlearning under perfect de-duplication and in federated unlearning, but this comes at the significant cost of retraining the entire model from scratch.



\subsection{Reinforcement Unlearning Study}
\noindent\textbf{Reinforcement Unlearning with DQN.} Given the significant differences between reinforcement unlearning and conventional machine unlearning, their evaluation metrics also take on distinct meanings. In reinforcement unlearning, model fidelity refers to the average reward that the agent receives in the remaining environments, while test performance denotes the average reward the agent achieves in a new testing environment. Unlearning efficacy quantifies the average reward that the agent attains in the unlearning environment, whereas unlearning impact represents the average reward received in the duplicated victim environment.

Table \ref{tab:RLDQN} presents the reinforcement unlearning outcomes in grid world using the DQN algorithm. Initially, the agent demonstrates commendable performance across all metrics. However, upon engaging in the unlearning process, several interesting phenomena emerge.
When the agent undergoes retraining from scratch, its performance remains largely unchanged compared to its pre-unlearning state. Notably, it continues to achieve a high reward in the unlearning environment, suggesting retained memory of the unlearning environment's features and indicating an unsuccessful unlearning attempt. This persistence in performance can be attributed to the agent's tendency to preserve the knowledge previously learned from the duplicated environment. 
For the other two unlearning methods, decremental RL-based and poisoning-based, they induce the agent's performance to decline in the unlearning environment, causing the agent to forget key features of the environment. Also, even when the duplicated victim environment deviates slightly from the unlearning environment, both decremental RL-based and poisoning-based methods still yield deteriorating performance in the unlearning environment, albeit exhibiting improved performance in the duplicated victim environment compared to the complete duplication scenario. This performance improvement may be attributed to the differences introduced in the near-duplicated environment. 
These results demonstrate that the adversary can achieve their objective under the decremental RL-based and poisoning-based methods, causing the agent to perform poorly in the victim environment, $\mathcal{M}_V$. However, the adversary’s objective is less achievable under the retraining-from-scratch approach. Despite this, retraining from scratch leads to high agent performance in the adversary's environment, $\mathcal{M}_A$, giving the adversary confidence to challenge the model owner by claiming that $\mathcal{M}_A$ was not properly unlearned.

\begin{table}[!ht]\scriptsize
	\centering
    \vspace{-2mm}
	\caption{Reinforcement unlearning results with DQN.}
\begin{tabular} {ccccc} 
\toprule\rowcolor{gray!40}
  & Model fidelity & \makecell[c]{Test\\performance} & \makecell[c]{Unlearn\\efficacy} & \makecell[c]{Unlearn\\impact} \\
 \midrule
Before unlearn  & $53.91$ & $48.35$ & $53.42$ & $53.55$ \\\hline \rowcolor{gray!40}
Retrain & $52.45$ & $44.13$ & $52.05$ & $51.53$ \\
Decremental RL+CD & $52.72$ & $43.36$ & $41.38$ & $42.21$ \\ \rowcolor{gray!40}
Decremental RL+ND & $53.66$ & $43.16$ & $43.62$ & $47.55$ \\
Poisoning+CD & $52.63$ & $44.95$ & $40.25$ & $41.45$ \\ \rowcolor{gray!40}
Poisoning+ND & $53.51$ & $45.32$ & $42.14$ & $48.67$ \\
\bottomrule
\end{tabular}
\vspace{-2mm}
	\label{tab:RLDQN}
\end{table}

\vspace{2mm}
\noindent\textbf{Reinforcement Unlearning with DDPG.} To modify the model architecture, we transition from using the DQN algorithm to the DDPG algorithm. The results are presented in Table \ref{tab:RLDDPG}. Generally, the trend observed is similar to the outcomes obtained with the DQN algorithm, as shown in Table \ref{tab:RLDQN}. However, there is a slight discrepancy in performance between the two algorithms. Specifically, the overall performance achieved using DDPG is slightly inferior to that achieved with DQN. This discrepancy may be attributed to differences in how the two algorithms handle the reinforcement learning task. For instance, the DDPG algorithm relies on a deterministic policy, which may lead to suboptimal exploration compared to the stochastic policy employed by DQN. 

\begin{table}[!ht]\scriptsize
	\centering
    \vspace{-2mm}
	\caption{Reinforcement unlearning results with DDPG.}
\begin{tabular} {ccccc} 
\toprule\rowcolor{gray!40}
  & Model fidelity & \makecell[c]{Test\\performance} & \makecell[c]{Unlearn\\efficacy} & \makecell[c]{Unlearn\\impact} \\
 \midrule
Before unlearn  & $49.47$ & $45.32$ & $48.94$ & $48.48$ \\\hline \rowcolor{gray!40}
Retrain & $48.88$ & $42.62$ & $45.42$ & $45.15$ \\
Decremental RL+CD & $48.26$ & $39.88$ & $41.44$ & $41.24$ \\\rowcolor{gray!40}
Decremental RL+ND & $49.01$ & $41.56$ & $42.45$ & $46.23$ \\
Poisoning+CD & $48.12$ & $39.15$ & $41.12$ & $41.52$ \\\rowcolor{gray!40}
Poisoning+ND & $48.97$ & $42.34$ & $42.45$ & $45.67$ \\
\bottomrule
\end{tabular}
\vspace{-2mm}
	\label{tab:RLDDPG}
\end{table}

\vspace{2mm}
\noindent\textbf{Reinforcement Unlearning with Two Environments Unlearned.} The results of unlearning two environments are depicted in Table \ref{tab:RL2Environments}. Notably, the overall performance trend remains consistent with that observed when unlearning a single environment. This similarity in performance across different numbers of environments may be attributed to the agent's ability to adapt and generalize across environments, allowing it to retain learned knowledge and strategies, even when faced with multiple instances of unlearning.

\begin{table}[!ht]\scriptsize
	\centering
    \vspace{-2mm}
	\caption{Reinforcement unlearning results with two unlearning environments.}
\begin{tabular} {ccccc} 
\toprule\rowcolor{gray!40}
  & Model fidelity & \makecell[c]{Test\\performance} & \makecell[c]{Unlearn\\efficacy} & \makecell[c]{Unlearn\\impact} \\
 \midrule
Before unlearn  & $53.24$ & $48.25$ & $52.62$ & $52.88$ \\\hline\rowcolor{gray!40}
Retrain & $52.23$ & $44.97$ & $52.31$ & $52.21$ \\
Decremental RL+CD & $52.14$ & $43.16$ & $42.34$ & $42.34$ \\\rowcolor{gray!40}
Decremental RL+ND & $53.21$ & $44.32$ & $43.57$ & $48.56$ \\
Poisoning+CD & $52.52$ & $43.25$ & $41.25$ & $42.15$ \\\rowcolor{gray!40}
Poisoning+ND & $53.22$ & $44.51$ & $43.24$ & $48.76$ \\
\bottomrule
\end{tabular}
\vspace{-2mm}
	\label{tab:RL2Environments}
\end{table}

\vspace{2mm}
\noindent\textbf{Reinforcement Unlearning with Defense.} Tables \ref{tab:RLDQNDefense1} and \ref{tab:RLDDPGDefense1} present the unlearning results using DQN and DDPG, respectively, under the first countermeasure: unlearning both the requested environment and the detected duplicate with the perfect de-duplication strategy. Under both the decremental RL-based and poisoning-based unlearning methods, the adversary successfully achieves their objective, causing the agent to perform poorly in the victim environment (i.e., low unlearning impact), even with the countermeasure applied.
The reason for this phenomenon is interesting. The first countermeasure unlearns both the requested environment and the detected duplicate. Since the detected duplicate is the victim environment itself, the countermeasure unintentionally assists the adversary by ensuring that the unlearning process negatively affects the agent’s performance in the victim environment. This aligns with the goals of the decremental RL-based and poisoning-based methods, which are designed to degrade the agent’s performance in the unlearning environment.

By contrast, the retraining-from-scratch method proves effective in resisting the adversary when DQN is used to train the agent (Table \ref{tab:RLDQNDefense1}). This is because retraining starts the learning process anew, and the victim environment is included in the retraining process, allowing the agent to regain high performance in the victim environment. However, when DDPG is used to train the agent, the retraining approach fails to resist the adversary (Table \ref{tab:RLDDPGDefense1}). This failure arises because DDPG’s training dynamics involve continuous action spaces and policy optimization, making it more challenging to fully recover the agent's capabilities in the victim environment. Additionally, the high complexity of DDPG’s parameter space may limit the agent's ability to fully regain the nuanced knowledge required for optimal performance in the victim environment. However, it is important to note that when using DQN, the agent also performs well in the unlearning environment. While this demonstrates the model's generalizability, it may inadvertently give the adversary confidence to challenge the success of unlearning, as the agent's retained performance could be interpreted as a failure to fully unlearn the environment.

\begin{table}[!ht]\scriptsize
	\centering
    \vspace{-2mm}
	\caption{Reinforcement unlearning results with DQN using the first countermeasure: unlearning both the requested environment and the detected duplicate under the perfect de-duplication strategy.}
\begin{tabular} {ccccc} 
\toprule\rowcolor{gray!40}
  & Model fidelity & \makecell[c]{Test\\performance} & \makecell[c]{Unlearn\\efficacy} & \makecell[c]{Unlearn\\impact} \\
 \midrule
Before unlearn  & $53.91$ & $48.35$ & $53.42$ & $53.55$ \\\hline \rowcolor{gray!40}
Retrain & $52.44$ & $44.32$ & $51.05$ & $51.54$\\ 
Decremental RL & $51.85$ & $42.92$ & $42.98$ & $43.12$ \\ \rowcolor{gray!40}
Poisoning & $50.23$ & $43.11$ & $42.35$ & $42.11$ \\
\bottomrule
\end{tabular}
\vspace{-2mm}
	\label{tab:RLDQNDefense1}
\end{table}

\begin{table}[!ht]\scriptsize
	\centering
    \vspace{-3mm}
	\caption{Reinforcement unlearning results with DDPG using the first countermeasure under the perfect de-duplication strategy.}
\begin{tabular} {ccccc} 
\toprule\rowcolor{gray!40}
  & Model fidelity & \makecell[c]{Test\\performance} & \makecell[c]{Unlearn\\efficacy} & \makecell[c]{Unlearn\\impact} \\
 \midrule
Before unlearn  & $53.91$ & $48.35$ & $53.42$ & $53.55$ \\\hline \rowcolor{gray!40}
Retrain & $48.45$ & $44.50$ & $43.25$ & $42.50$ \\ 
Decremental RL & $47.91$ & $40.98$ & $40.67$ & $40.78$ \\ \rowcolor{gray!40} 
Poisoning & $47.42$ & $40.32$ & $39.78$ & $40.09$ \\
\bottomrule
\end{tabular}
\vspace{-2mm}
	\label{tab:RLDDPGDefense1}
\end{table}

Tables \ref{tab:RLDQNDefense2} and \ref{tab:RLDDPGDefense2}, which evaluate the second countermeasure: removing the victim environment and unlearning the requested environment, show results similar to those observed with the first countermeasure. Under both the decremental RL-based and poisoning-based unlearning methods, the outcomes remain largely unaffected by whether the victim environment is removed. This is because these methods directly target the unlearning environment and aim to degrade the agent's performance within it. Consequently, the inclusion or exclusion of the victim environment has little impact on their effectiveness.

For the retraining-from-scratch method, removing the victim environment excludes it from the retraining process. However, when DQN is used to train the agent, the retrained model demonstrates strong generalizability, allowing it to perform well in the victim environment despite its absence during retraining. This behavior is akin to the training versus testing data scenario in standard machine learning, where a well-trained model can generalize effectively to unseen data. By contrast, the generalizability is less pronounced with DDPG due to its reliance on continuous action spaces and policy optimization, which require nuanced representations of the training environments for effective performance.

\begin{table}[!ht]\scriptsize
	\centering
    \vspace{-3mm}
	\caption{Reinforcement unlearning results with DQN using the second countermeasure: removing the victim environment and unlearning the requested environment under the perfect de-duplication strategy.}
\begin{tabular} {ccccc} 
\toprule\rowcolor{gray!40}
  & Model fidelity & \makecell[c]{Test\\performance} & \makecell[c]{Unlearn\\efficacy} & \makecell[c]{Unlearn\\impact} \\
 \midrule
Before unlearn  & $53.91$ & $48.35$ & $53.42$ & $53.55$ \\\hline \rowcolor{gray!40}
Retrain & $51.89$ & $43.88$ & $50.65$ & $51.22$ \\ 
Decremental RL & $52.24$ & $43.42$ & $43.55$ & $43.25$ \\ \rowcolor{gray!40} 
Poisoning & $51.49$ & $43.55$ & $42.85$ & $42.62$ \\ 
\bottomrule
\end{tabular}
\vspace{-2mm}
	\label{tab:RLDQNDefense2}
\end{table}

\begin{table}[!ht]\scriptsize
	\centering
    \vspace{-2mm}
	\caption{Reinforcement unlearning results with DDPG using the second countermeasure under the perfect de-duplication strategy.}
\begin{tabular} {ccccc} 
\toprule\rowcolor{gray!40}
  & Model fidelity & \makecell[c]{Test\\performance} & \makecell[c]{Unlearn\\efficacy} & \makecell[c]{Unlearn\\impact} \\
 \midrule
Before unlearn  & $53.91$ & $48.35$ & $53.42$ & $53.55$ \\\hline \rowcolor{gray!40}
Retrain & $48.02$ & $44.12$ & $42.92$ & $41.32$ \\ 
Decremental RL & $48.42$ & $41.53$ & $41.15$ & $41.33$ \\ \rowcolor{gray!40} 
Poisoning & $47.95$ & $41.04$ & $40.01$ & $41.10$ \\ 
\bottomrule
\end{tabular}
\vspace{-3mm}
	\label{tab:RLDDPGDefense2}
\end{table}

\vspace{1mm}
\noindent\textbf{Summary.} In reinforcement unlearning, the adversary effectively achieves their goal of degrading the agent’s performance in the victim environment using both the decremental RL-based and poisoning-based methods, even when countermeasures are applied. The adversary’s goal is less achievable under the retraining-from-scratch approach, especially when DQN is used to train the agent. However, the agent continues to perform well in the adversary’s environment, enabling the adversary to challenge the model owner’s unlearning success.

\vspace{-2mm}
\section{Related Work}
\vspace{-2mm}
\noindent\textbf{Data Duplication and De-duplication.} 
Current studies primarily center on two objectives: understanding the impact of data duplication on model performance and addressing security concerns.
For instance, Lee et al. \cite{Lee22ACL} revealed that applying de-duplication to the training datasets of language models significantly reduces the risk of generating memorized text and requires fewer training steps to achieve comparable model accuracy. Similarly, Carlini et al. \cite{Carlini21USENIX} and Kandpal et al. \cite{Kandpal22ICML} found that data duplication can assist adversaries in generating sequences from a trained model and identifying which sequences are memorized from the training set. 
Moreover, duplication can also be employed to undermine model security. Rakin et al. \cite{Rakin21} utilized adversarial weight duplication to inject specific DNN weight packages during data transmission, aiming to hijack the DNN function of the victim tenant in a cloud computing scenario.
However, existing research primarily focuses on the learning process of models, with the unlearning process often overlooked.

\vspace{1mm}
\noindent\textbf{Machine Unlearning.} 
Current research primarily delves into the unlearning process and the verification of unlearning results \cite{Xu23}. For example, Bourtoule et al. \cite{Bourtoule21} proposed SISA (Sharded, Isolated, Sliced, and Aggregated) training by randomly partitioning the training set into multiple shards and training a constituent model for each shard. In the event of an unlearning request, the model provider only needs to retrain the corresponding shard model. Subsequently, Warnecke et al. \cite{Warnecke23} shifted the focus of unlearning research from removing samples to removing features and labels, employing the concept of influence functions in their approach. Recently, Thudi et al. \cite{Thudi23} argued that unlearning cannot be proven solely by training the model on the unlearned data, as done in \cite{Ginart19,Guo20}, but should only be defined at the level of the algorithms used for learning and unlearning.

Recent research has also shed light on vulnerabilities in machine unlearning, elucidating critical trends such as over-unlearning and privacy leakage. Over-unlearning occurs when the unlearning process inadvertently removes more information than intended \cite{Hu2024NDSS}, potentially leading to degraded model performance or unintended data loss. Conversely, privacy leakage \cite{Hu2024SP} poses another formidable challenge, wherein adversaries can reconstruct sensitive information from unlearned data by analyzing discrepancies between the output of the original model and the unlearned model.

Investigations into federated unlearning and reinforcement unlearning have revealed distinct features. For instance, federated unlearning focuses more on gradient-level unlearning and can be conducted to remove an entire client \cite{Wu22Network}. 
Reinforcement unlearning is carried out at the object-level, aiming to remove the influence of an entire environment \cite{Ye25NDSS}.

However, the issue of data duplication has not been considered during the unlearning process. This paper fills this gap by exploring the impact of duplicate data on machine unlearning, providing novel insights into the overlooked aspects of model refinement and adaptation after training.

\vspace{1mm}
\noindent\textbf{Data Poisoning Attacks.} Adversarially injecting duplicate data shares similarities with data poisoning attacks \cite{Shan22USENIX,Marchant22AAAI}, particularly in the aspect of modifying the training dataset. Nonetheless, these two techniques differ significantly. The primary objective of data duplication is to introduce redundancy of existing data instances into the training set, either intentionally or unintentionally. This redundancy, by itself, is not considered a poison. In some scenarios, duplicate data is intentionally introduced to improve model performance. For instance, in image data augmentation, strategies such as rotation and scaling are often employed to generate additional images that are near-duplicates of their original counterparts. In contrast, data poisoning aims to manipulate the training set by strategically injecting malicious samples, often with the intent of deceiving the model during training. 

\vspace{-3mm}
\section{Conclusion}
\vspace{-1mm}
This paper has provided a comprehensive exploration into the impact of data duplication on machine unlearning processes, spanning conventional, federated, and reinforcement learning. Our findings underscore the importance of considering duplication in unlearning methodologies, revealing instances where standard approaches may fail and highlighting the need for tailored strategies to mitigate the effects of duplication on model performance and privacy.
Future research directions include delving deeper into potential countermeasures and the development of advanced de-duplication techniques.

\vspace{-2mm}
\section*{Acknowledgments}
\vspace{-1mm}
This work is partially supported by the ARC projects LP220200808 and DP230100246. This work is also funded by the European Health and Digital Executive Agency (HADEA) within the project ``Understanding the individual host response against Hepatitis D Virus to develop a personalized approach for the management of hepatitis D'' (DSolve, grant agreement number 101057917) and the BMBF with the project ``Repräsentative, synthetische Gesundheitsdaten mit starken Privatsphärengarantien'' (PriSyn, 16KISAO29K).

\vspace{-2mm}
\section*{Ethics Considerations}
\vspace{-1mm}
By highlighting the challenges and complexities associated with data duplication in machine unlearning, our research contributes to the development of more secure and robust machine learning models. We aim to enhance the awareness of potential vulnerabilities that can arise from improper handling of duplicated data, especially in scenarios involving adversarial manipulation. Also, in line with ethical research practices, we have chosen not to publish any specific data embeddings that could potentially be misused for malicious purposes. 

\vspace{-2mm}
\section*{Open Science Statement}
\vspace{-1mm}
We reproduce the state-of-the-art baselines by utilizing their official repositories on GitHub to ensure the validity of our comparisons. We will release our code, data duplication techniques, and experimental setups to facilitate further research in machine unlearning and to enable other researchers to validate our findings, build upon our work, and develop more effective unlearning and de-duplication strategies.




\vspace{-2mm}
\bibliographystyle{plain}
\bibliography{references}


\vspace{-4mm}
\section*{Appendix}
\vspace{-2mm}
\setcounter{section}{0}
\renewcommand{\appendixname}{Appendix~\Alph{section}}
\section{Proof of Lemma \ref{lem:state distribution} and Theorem \ref{thm:reward}}
\vspace{-2mm}

\begin{proof}[Proof of Lemma \ref{lem:state distribution}]
\begin{equation}\nonumber
\begin{aligned}
    &d^{\mathcal{T}',\pi}_{\mu}-d^{\mathcal{T},\pi}_{\mu}\\
    &=\gamma\int_{\mathcal{S}}d^{\mathcal{T}',\pi}_{\mu}(s')\mathcal{T}'^{\pi}(s|s')\mathrm{d}s'-\gamma\int_{\mathcal{S}}d^{\mathcal{T},\pi}_{\mu}(s')\mathcal{T}^{\pi}(s|s')\mathrm{d}s'\\
    &=\gamma\int_{\mathcal{S}}d^{\mathcal{T}',\pi}_{\mu}(s')\mathcal{T}'^{\pi}(s|s')\mathrm{d}s'-\gamma\int_{\mathcal{S}}d^{\mathcal{T},\pi}_{\mu}(s')\mathcal{T}^{\pi}(s|s')\mathrm{d}s'\\
    &+\gamma\int_{\mathcal{S}}d^{\mathcal{T}',\pi}_{\mu}(s')\mathcal{T}^{\pi}(s|s')\mathrm{d}s'-\gamma\int_{\mathcal{S}}d^{\mathcal{T}',\pi}_{\mu}(s')\mathcal{T}^{\pi}(s|s')\mathrm{d}s'\\
    &=\gamma\int_{\mathcal{S}}d^{\mathcal{T}',\pi}_{\mu}(s')(\mathcal{T}'^{\pi}(s|s')-\mathcal{T}^{\pi}(s|s'))\mathrm{d}s'\\
    &+\gamma\int_{\mathcal{S}}(d^{\mathcal{T}',\pi}_{\mu}(s')-d^{\mathcal{T},\pi}_{\mu}(s'))\mathcal{T}^{\pi}(s|s')\mathrm{d}s'.
\end{aligned}
\end{equation}
Then, we have:
\begin{equation}\nonumber
\begin{aligned}
    &||d^{\mathcal{T}',\pi}_{\mu}-d^{\mathcal{T},\pi}_{\mu}||_1\\
    &\leq\gamma\int_{\mathcal{S}}|\int_{\mathcal{S}}d^{\mathcal{T}',\pi}_{\mu}(s')(\mathcal{T}'^{\pi}(s|s')-\mathcal{T}^{\pi}(s|s'))\mathrm{d}s'|\mathrm{d}s\\
    &+\gamma\int_{\mathcal{S}}|\int_{\mathcal{S}}(d^{\mathcal{T}',\pi}_{\mu}(s')-d^{\mathcal{T},\pi}_{\mu}(s'))\mathcal{T}^{\pi}(s|s')\mathrm{d}s'|\mathrm{d}s\\
    &\leq\gamma\int_{\mathcal{S}}d^{\mathcal{T}',\pi}_{\mu}(s')\int_{\mathcal{S}}|\mathcal{T}'^{\pi}(s|s')-\mathcal{T}^{\pi}(s|s')|\mathrm{d}s\mathrm{d}s'\\
    &+\gamma\int_{\mathcal{S}}|d^{\mathcal{T}',\pi}_{\mu}(s')-d^{\mathcal{T},\pi}_{\mu}(s')|\int_{\mathcal{S}}\mathcal{T}^{\pi}(s|s')\mathrm{d}s\mathrm{d}s'\\
    &\leq\gamma\mathbb{E}_{s\sim d^{\mathcal{T},\pi}_{\mu}}||\mathcal{T}'^{\pi}(\cdot|s)-\mathcal{T}^{\pi}(\cdot|s)||_1+\gamma||d^{\mathcal{T}',\pi}_{\mu}-d^{\mathcal{T},\pi}_{\mu}||_1\\
    &\Rightarrow||d^{\mathcal{T}',\pi}_{\mu}-d^{\mathcal{T},\pi}_{\mu}||_1\leq\frac{\gamma}{1-\gamma}\mathbb{E}_{s\sim d^{\mathcal{T},\pi}_{\mu}}||\mathcal{T}'^{\pi}(\cdot|s)-\mathcal{T}^{\pi}(\cdot|s)||_1,
\end{aligned}
\end{equation}
where the first inequality is established by applying the definition of the norm; the second inequality is derived through the subadditivity property of the norm; and the third inequality is attained by noting that $\int_{\mathcal{S}}\mathcal{T}^{\pi}(s|s')\mathrm{d}s=1$.
\end{proof}

\begin{proof}[Proof of Theorem \ref{thm:reward}]
\begin{equation}\nonumber
\begin{aligned}
    &J^{\mathcal{T}',\pi}_{\mu}-J^{\mathcal{T},\pi}_{\mu}=\frac{1}{1-\gamma}\int_{\mathcal{S}}d^{\mathcal{T}',\pi}_{\mu}(s)\int_{\mathcal{A}}\pi(a|s)r(s,a){\rm d}a{\rm d}s\\
    &-\frac{1}{1-\gamma}\int_{\mathcal{S}}d^{\mathcal{T},\pi}_{\mu}(s)\int_{\mathcal{A}}\pi(a|s)r(s,a){\rm d}a{\rm d}s\\
    &=\frac{1}{1-\gamma}\int_{\mathcal{S}}(d^{\mathcal{T}',\pi}_{\mu}(s)-d^{\mathcal{T},\pi}_{\mu}(s))\int_{\mathcal{A}}\pi(a|s)r(s,a){\rm d}a{\rm d}s\\
    &\leq\frac{1}{1-\gamma}\int_{\mathcal{S}}|d^{\mathcal{T}',\pi}_{\mu}(s)-d^{\mathcal{T},\pi}_{\mu}(s)|\int_{\mathcal{A}}\pi(a|s)r(s,a){\rm d}a{\rm d}s\\
    &\leq\frac{1}{1-\gamma}||d^{\mathcal{T}',\pi}_{\mu}-d^{\mathcal{T},\pi}_{\mu}||_1\max_s\mathbb{E}_{a\sim\mathcal{A}}r(s,a)\\
    &\leq\frac{\gamma}{(1-\gamma)^2}\mathbb{E}_{s\sim d^{\mathcal{T},\pi}_{\mu}}||\mathcal{T}'^{\pi}(\cdot|s)-\mathcal{T}^{\pi}(\cdot|s)||_1\max_s\mathbb{E}_{a\sim\mathcal{A}}r(s,a),
\end{aligned}
\end{equation}
where the first inequality is obtained by leveraging the properties of integrals; the second inequality is derived by noting that $\mathbb{E}_{a\sim\mathcal{A}}r(s,a)=\int_{\mathcal{A}}\pi(a|s)r(s,a){\rm d}a$; and the third inequality is established by applying the conclusion of Lemma \ref{lem:state distribution}.
\end{proof}

\end{document}